%% file: infoshare_single.tex
\documentclass[11pt, journal, draftclsnofoot, onecolumn]{IEEEtran}
\usepackage{amsthm}
\usepackage{amsfonts}
\usepackage{amssymb}
\usepackage{hyperref}
\usepackage{graphicx}
\usepackage{enumerate}
\usepackage{amsmath}
\usepackage{color}

\input{macros.tex}

\begin{document}
%
\title{Information Sharing in Networks \\of Strategic Agents}
%
%
%

\author{Jie~Xu,
        Yangbo~Song,
        and~Mihaela~van~der~Schaar,~\IEEEmembership{Fellow,~IEEE}
\thanks{Jie Xu and Mihaela van der Schaar are with the Dept. of Electrical Engineering, University of California, Los Angeles (UCLA). Emails: jiexu@ucla.edu, mihaela@ee.ucla.edu.
}
\thanks{Yangbo Song is with the Dept. of Economics, University of California, Los Angeles (UCLA). Email: darcy07@ucla.edu.}}

\maketitle

\begin{abstract}
To ensure that social networks (e.g. opinion consensus, cooperative estimation, distributed learning and adaptation etc.) proliferate and efficiently operate, the participating agents need to collaborate with each other by repeatedly sharing information. However, sharing information is often costly for the agents while resulting in no direct immediate benefit for them. Hence, lacking incentives to collaborate, strategic agents who aim to maximize their own individual utilities will withhold rather than share information, leading to inefficient operation or even collapse of networks. In this paper, we develop a systematic framework for designing \textit{distributed} rating protocols aimed at incentivizing the strategic agents to collaborate with each other by sharing information. The proposed incentive protocols exploit the ongoing nature of the agents' interactions to assign ratings and through them, determine future rewards and punishments: agents that have behaved as directed enjoy high ratings -- and hence greater future access to the information of others; agents that have not behaved as directed enjoy low ratings -- and hence less future access to the information of others. Unlike existing rating protocols, the proposed protocol operates in a distributed manner, online, and takes into consideration the underlying interconnectivity of agents as well as their heterogeneity. We prove that in many deployment scenarios the price of anarchy (PoA) obtained by adopting the proposed rating protocols is one. In settings in which the PoA is larger than one, we show that the proposed rating protocol still significantly outperforms existing incentive mechanisms such as Tit-for-Tat. Importantly, the proposed rating protocols can also operate efficiently in deployment scenarios where the strategic agents interact over time-varying network topologies where new agents join the network over time.
\end{abstract}


%
\IEEEpeerreviewmaketitle

\section{Introduction}
In recent years, extensive research efforts have been devoted to studying cooperative networks where agents interact with each other over a topology repeatedly, by sharing information such as measurements, estimates, beliefs, or opinions. Such networks involve various levels of coordinated behavior among agents in order to solve important tasks in an efficient and distributed manner such as target tracking, object detection, resource allocation, learning, inference, and estimation. Collaboration among the agents via repeated information sharing is critical for the enhanced performance and robustness of the distributed solution, as already demonstrated in various insightful studies on social learning in multi-agent networks \cite{KrishnamurthyA}-\cite{Jadbabaie}, belief consensus in social networks \cite{Chamley}\cite{Acemouglu}, distributed optimization in resource allocation problems \cite{Tsitsiklis}-\cite{Dimakis} and in the diffusion of information for adaptation and learning purposes \cite{Chen}-\cite{Sayed}. However, in many scenarios, participating in the cooperative process entails costs to the agents, such as the cost of producing, transmitting, and sharing information with their neighbors. In these situations, the cost of sharing information may outweigh the benefit of cooperation and agents may not see an immediate benefit to being cooperative. For networks where agents are strategic, meaning that they aim to maximize their own utilities by strategically choosing their actions, the agents will choose to participate in the collaborative process only if they believe this action is beneficial to their current and long-term interests. Absent incentives for collaboration, these networks will work inefficiently or can even collapse \cite{Lucky}. A distinct feature of the network under consideration is that agents' incentives can be coupled in a possibly extremely complex way due to the underlying topology. Thus, a key challenge to ensure the survivability and efficient operation of networks in the presence of selfish agents is the design of incentive schemes that adapt to the network topology and encourage the agents' cooperation in accordance with the network objective.

We propose to resolve the above incentive problem by exploiting the repeated interactions among agents to enable social reciprocation, by deploying a \textit{distributed} rating protocol. Such rating protocols are designed and implemented in a distributed manner and are tailored to the underlying topologies. The rating protocol, via the (non-strategic) Social Network Interface (SNI)\footnote{ For example, the SNIs are tamper-proof software/hardware modules that can communicate with other SNIs in the neighborhood. However, they do not communicate with a central entity and hence, they are also distributed.} with which each agent is equipped, recommends (online and in a distributed way) to every agent how much information they should share with their neighbors depending on each neighbor's current rating according to the network topology. We refer to this recommendation as the \textit{recommended strategy}. Importantly, the protocol has to be designed in such a way that this recommendation is incentive-compatible, meaning that agents have incentives to follow it. (We will later define a more formal version of ``incentive-compatibility''.)  In each period, agents have the freedom to decide how much information they should share with each of their neighbors. Their decision may comply or not with the strategy recommended (i.e. agents may follow or deviate from this strategy).  The agent's rating is then increased/decreased by the SNI based on its current rating, and whether it has followed/deviated from the recommended strategy. We refer to this as the \textit{rating update} rule. High-rated agents will be rewarded -- the protocol recommends more information sharing by their neighbors and hence they receive more benefit in the future; low-rated agents will be punished -- the protocol recommends less information sharing by their neighbors and hence, they receive less benefit in the future.

Next, we highlight two distinct features of the networks under consideration and the resulting key challenges for designing rating protocols for agents to cooperate. The first feature is that agents interact over an underlying topology. This is in stark contrast with existing works in repeated games relying on social reciprocation which assume that the agents are randomly matched \cite{Kandori}\cite{Zhang}\cite{Xu}. In this paper, agents' incentives are coupled in a complex manner since their utilities depend on the behavior of the other agents with which they are interconnected. Since agents have different neighbors, their incentives can also be very diverse. A recommended strategy and rating update rule may provide sufficient incentives for some agents to follow but may fail in incentivizing others. In the worst cases, even a single agent deviating from the recommended strategy may cause a ``chain effect'' where eventually all agents deviate, leading to the collapse of the network. Hence, the rating protocol must be designed to adapt to the specific network topology.

The second feature is that the networks under consideration are distributed and hence, they are informationally decentralized, in the sense that (i) communication can only occur between neighboring agents (and SNIs) and (ii) there is no central planner that can monitor the entire network and communicate to the individual agents information about each agent's behavior (e.g. its compliance with the recommended strategy in the past, its rating etc.). Decentralization prevents rating protocols proposed in prior works \cite{Zhang}\cite{Xu} from being applicable in the considered scenarios since they are designed and implemented in a centralized manner. Therefore, a new distributed rating protocol needs to be developed which can operate successfully in an informationally-decentralized network.

The remaining part of this paper is organized as follows. In Section II, we review related works and existing solutions, and highlight the key differences to this work. Section III outlines the system model and formulates the protocol design problem. The structure of the rating protocol is unraveled in Section IV. In Section V, we design the optimal rating protocol to maximize the social welfare. The performance of the optimal design is then analyzed in Section VI. Section VII studies the rating protocol design in a class of time-varying topologies. Section VIII provides numerical results to highlight the features of the proposal. Finally, we conclude this paper in Section XI.

\section{Related Works}

\begin{table*}[t]
\centerline{\includegraphics[scale = 1]{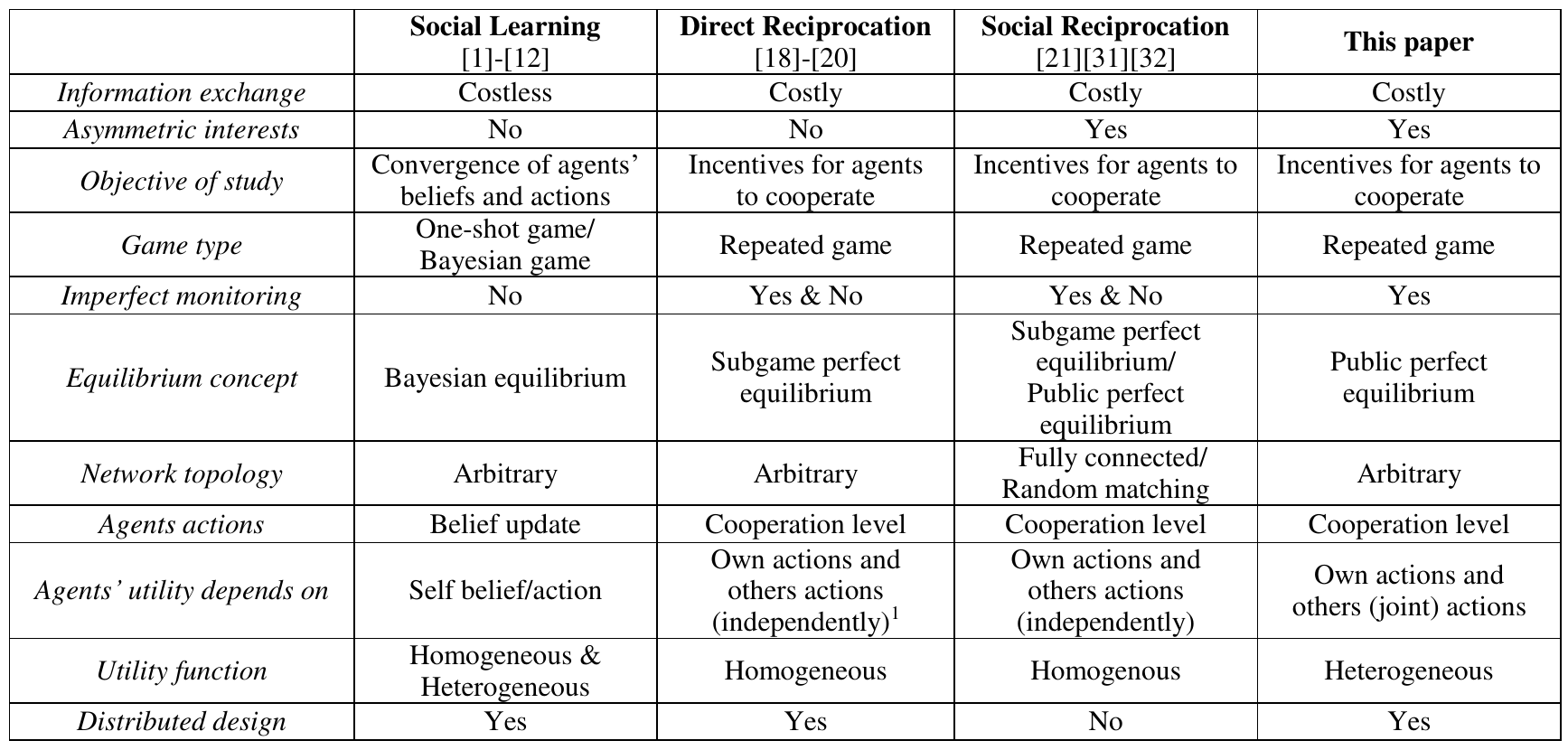}}
\caption{Comparison with existing works. }
\label{com_existing}
\vspace{-10pt}
\end{table*}

Collaboration among the agents via repeated information sharing is critical for the enhanced performance and robustness of various types of social networks \cite{KrishnamurthyA}-\cite{Sayed}. The main focus of this literature is on determining the resulting network performance if agents repeatedly share and process information in various ways. However, absent incentives and in the presence of selfish agents, these networks will work inefficiently or can even collapse \cite{Lucky}. Thus, the main focus of the current paper is how to incentivize strategic agents to share such information such that this type of social networks can operate efficiently.

A variety of incentive schemes has been proposed to encourage cooperation among agents (see e.g. \cite{Park} for a review of different game theoretic solutions). Two popular incentive schemes are pricing and differential service. Pricing schemes \cite{Bergemann}\cite{MacKie-Mason} use payments to reward and punish individuals for their behavior. However, they often require complex accounting and monitoring infrastructure, which introduce substantial communication and computation overhead. Differential service schemes, on the other hand, reward and punish individuals by providing differential services depending on their behavior. Differential services can be provided by the network operator. However, in many distributed information sharing networks, such a centralized network operator does not exist. Alternatively, differential services can also be provided by the other agents participating in the network since agents in the considered applications derive their utilities from their interactions with other agents \cite{Axelrod}-\cite{Jackson}\cite{Zhang}\cite{Xu}. Such incentive schemes are based on the principle of reciprocity and can be classified into direct (personal) reciprocation and social reciprocation. In direct (personal) reciprocation schemes (e.g. the widely adopted Tit-for-Tat strategy \cite{Axelrod}-\cite{Milan}), the behavior of an individual agent toward another is based on its personal experience with that agent. However, they only work when two interacting agents have common interests. In social reciprocation schemes \cite{Song}-\cite{Jackson}\cite{Zhang}\cite{Xu}, individual agents obtain some (public) information about other individuals (e.g. their ratings) and decide their behavior toward other agents based on this information.

Incentive mechanisms based on social reciprocation are often studied using the familiar framework of repeated games. In \cite{Song}, the information sharing game is studied in a narrower context of cooperative spectrum sensing and various simple strategies are investigated. Agents are assumed to be able to communicate and share information with all other agents, effectively forming a clique topology where the agents' knowledge of the network is complete and symmetric. However, such an assumption rarely holds in distributed networks where, instead, agents may interact over arbitrary topologies and have incomplete and asymmetric knowledge of the entire network. In such scenarios, simple strategies proposed in \cite{Song} will fail to work and the incentives design becomes significantly more challenging.

Contagion strategies on networks \cite{Kandori}-\cite{Jackson} are proposed as a simple strategy to provide incentives for agents to cooperate. However, such strategies do not perform well if monitoring is imperfect since any single error can lead to a network collapse. Even if certain forms of forgiveness are introduced, contagion strategies are shown to be effective only in very specific topologies \cite{Ali}\cite{Jackson}. It is still extremely difficult, if not impossible, to design efficient forgiving schemes in distributed networks with arbitrary topologies since agents will have difficulty in conditioning their actions on history, e.g. whether they are in the contagion phase or the forgiving phase, due to the asymmetric and incomplete knowledge.

Rating/reputation mechanisms are proposed as another promising solution to implement social reciprocation. Much of the existing work on reputation mechanism is concerned with practical implementation details such as effective information gathering techniques \cite{Kamvar} or determining the impact of reputation on a seller's prices and sales \cite{Ba}\cite{Resnick}. The few works providing theoretical results on rating protocol design consider either one (or a few) long-lived agent(s) interacting with many short-lived agents \cite{Dellarocas}-\cite{Zacharia} or anonymous, homogeneous and unconnected agents selected to interact with each other using random matching \cite{Kandori}\cite{Zhang}\cite{Xu}. Importantly, few of the prior works consider the design of such rating protocols for networks where agents interact over an underlying topology which leads to extremely complexly-coupled interactions among agents. Moreover, the distributed nature of the considered information sharing networks imposes unique challenges for the rating protocol design and implementation which are not addressed in prior works \cite{Zhang}\cite{Xu}.

In Table 1, we compare the current paper with existing works on social learning and incentive schemes based on direct reciprocation and social reciprocation.

\section{System Model}
We consider a network of $N$ agents, indexed by $\{1,2,...,N\} = {\mathcal N}$. Agents are connected subject to an underlying topology $G=\{ g_{ij} \} _{i,j\in {\rm {\mathcal N}}} $ with $g_{ij} =g_{ji} =1$ (here we consider undirected connection) representing agent $i$ and $j$ being connected (e.g. there is a communication channel between them) and $g_{ij} =g_{ji} =0$ otherwise. Moreover, we set $g_{ii} =0$. We say that agent $i$ and agent $j$ are neighbors if they are connected. For now we assume a fixed topology $G$ but certain types of time-varying topologies are allowed in our framework and this will be discussed in detail in Section VII.

Time is divided into discrete periods. In each time period, each agent $i$ decides an information sharing action with respect to each of its neighbors $j$, denoted by $a_{ij} \in [0,1]$. For example, $a_{ij} $ can represent the information sharing effort by agent $i$ with agent $j$. We collect the actions of agent $i$ with respect to all its neighbors in the notation $\a_{i} =\{ a_{ij} \} _{j:g_{ij} =1} $. Denote $\a=(\a_{1} ,...,\a_{N} )$ as the action profile of all agents and $\a_{-i} =(\a_{1} ,...,\a_{i-1} ,\a_{i+1} ,...,\a_{N} )$ as the action profile of agents except $i$. Let ${\rm {\mathcal A}}_{i} =[0,1]^{m_{i} } $ be the action space of agent $i$ where $m_{i} =\sum _{j}g_{ij}  $. Let ${\rm {\mathcal A}}=\times _{i\in {\rm {\mathcal N}}} {\rm {\mathcal A}}_{i} $ be the action space of all agents.

Agents obtain benefits from neighbors' sharing actions. We denote the actions of agent $i$'s neighbors with respect to agent $i$ by $\hat{\a}_{i} {\rm =}\{ a_{ji} \} _{j:g_{ij} =1} $ and let $b_{i} (\hat{\a}_{i} )$ be the benefit that agent $i$ obtains from its neighbors \footnote{ In principle, an agent can obtain benefits from the information sharing over indirect links relayed by its neighbor. In this case, the action will also include the relaying action. }. Sharing information is costly and the cost $c_{i} (\a_{i} )$ depends on an agent $i$'s  own actions $\a_{i}$. Hence, given the action profile $\a$ of all agents, the utility of agent $i$ is
\begin{equation} \label{ZEqnNum432082}
u_{i} (\a)=b_{i} (\hat{\a}_{i} )-c_{i} (\a_{i} )
\end{equation}

We impose some constraints on the benefit and cost functions.

\textit{Assumption}: (1) For each $i$,  the benefit $b_{i} (\hat{\a}_{i} )$ is non-decreasing in each $a_{ij} ,\forall j:g_{ij} =1$ and is concave in $\hat{\a}_{i} $ (in other words, jointly concave in $a_{ji} ,\forall j:g_{ij} =1$). (2) For each $i$, the cost is linear in its sum actions, i.e. $c_{i} (\a_{i} )=\|\a_{i}\|_1 =\sum _{j:g_{ij} } a_{ij} $.

The above assumption states that (1) agents receive decreasing marginal benefits of information acquisition, which captures the fact that agents become more or less ``satiated'' when they possess sufficient information, in the sense that additional information would only generate little additional payoff; (2) the cost incurred by an agent is equal (or proportional) to the sum effort of collaboration with all its neighbors.

\subsection{Example: Cooperative Estimation}
We illustrate the generality of our formalism by showing how well-studied cooperative estimation problems \cite{Mishra}\cite{Unnikrishnan} can be cast into it. Consider that each agent observes in each period a noisy version of a time-varying underlying system parameter $s(t)$ of interest. Denote the observation of agent $i$ by $o_{i} (t)$. We assume that $o_{i} (t)=s(t)+ \eps_i(t)$, where the observation error $\eps_{i} (t)$ is i.i.d. Gaussian across agents and time with mean zero and variance $r^{2} $. Agents can exchange observations with their neighbors to obtain better estimations of the system parameter. Let $a_{ij} (t)$ be the transmission power spent by agent $i$. The higher the transmission power the larger probability that agent $j$ receives this additional observation from agent $i$. Agents can use various combination rules \cite{Chen} to obtain the final estimations. The expected mean square error (MSE) of agent $i$'s final estimation will depend on the actions of its neighbors, denoted by $MSE_{i} (\hat{\a}_{i} (t))$. If we define the MSE improvement as the benefit of agents, i.e. $b_{i} (\hat{\a}_{i} (t))=r^{2} -MSE(\hat{\a}_{i} (t))$, then the utility of agent $i$ in period $t$ given the received benefit and its incurred cost is $u_{i} (\a(t))=r^{2} -MSE_{i} (\hat{\a}_{i} (t))-\a_{i} (t)$.

\subsection{Obedient Agents -- Benchmark}
Even though this paper focuses on strategic agents in information sharing networks, it is useful to first study how obedient agents (i.e. non-strategic agents who follow any prescribed strategy) interact in order to obtain a better understanding of the interactions and the achievable performance. The objective of the protocol designer in this benchmark case is to maximize the social welfare of the network, which is defined as the time-average sum utility of all agents, i.e.
\begin{equation} \label{ZEqnNum110798}
V=\mathop{\lim }\limits_{T\to \infty } \frac{1}{T} \sum _{t=0}^{\infty }\sum _{i}u_{i}   (\a(t))
\end{equation}
where $\a(t)$ is the action profile in period $t$. If agents are obedient, then the system designer can assign socially optimal actions, denoted by $\a^{opt} (t),\forall t$, to agents and then agents will simply take the actions prescribed by the system designer. Determining the socially optimal actions involves solving the following utility maximization problem \cite{Palomar}:
\begin{equation} \label{ZEqnNum918213}
\begin{array}{l} {\mathop{{\rm maximize}}\limits_{a} {\rm \; \; \; \; \; \; \; \; }V} \\ {{\rm subject\; to\; \; \; \; \; \; \; \; }a_{ij} (t)\in [0,1],\forall i,j:g_{ij} =1,\forall t} \end{array}
\end{equation}
This problem can be easily solved and any action profile $\a^{opt}$ that satisfies
\begin{equation} \label{ZEqnNum389088}
\hat{\a}_{i}^{opt} (t)\in \arg \max _{\hat{\a}} b_{i} (\hat{\a}_{i} (t))-\hat{\a}_{i} (t)
\end{equation}
is its solution. We denote the optimal social welfare by $V^{opt} $.

In a distributed network, there is no central planner that knows everything about the network (including the network size, topology and individual agents' utility functions) and can communicate to all agents. However, the structure of problem \eqref{ZEqnNum918213} lends itself to a fully decentralized implementation \cite{Rockafellar}: each SNI can compute the optimal actions for its neighbors by solving \eqref{ZEqnNum389088} and sending the solution to their neighboring SNIs. In this way, if all agents take the actions solved by the SNIs, the social welfare is maximized.

It is helpful to give an illustrative example of the optimal information sharing actions for agents connected using different topologies. We will revisit this example when we study strategic agents and show how incentives design and information sharing strategies are affected by the topologies.

\textit{Example}: (Ring and Star topologies) We consider a set of 4 agents performing cooperative estimation (as in Section III. A) over two fixed topologies -- a ring and a star. A possible approximation of the utility function of each agent $i$ when the uniform combination rule is used is $u_{i} (\a(t))=r^{2} -\frac{r^{2} }{1+\sum _{j:g_{ij} } a_{ji} } -\sum _{j:g_{ij} } a_{ij} $. We assume that the noise variance $r^{2} =4$. Figure \ref{ring-star1} illustrates the optimal actions in different topologies by solving \eqref{ZEqnNum918213}. In both topologies, the optimal social welfare is $V^{opt} =4$.

\begin{figure}
\centerline{\includegraphics[scale = 0.7]{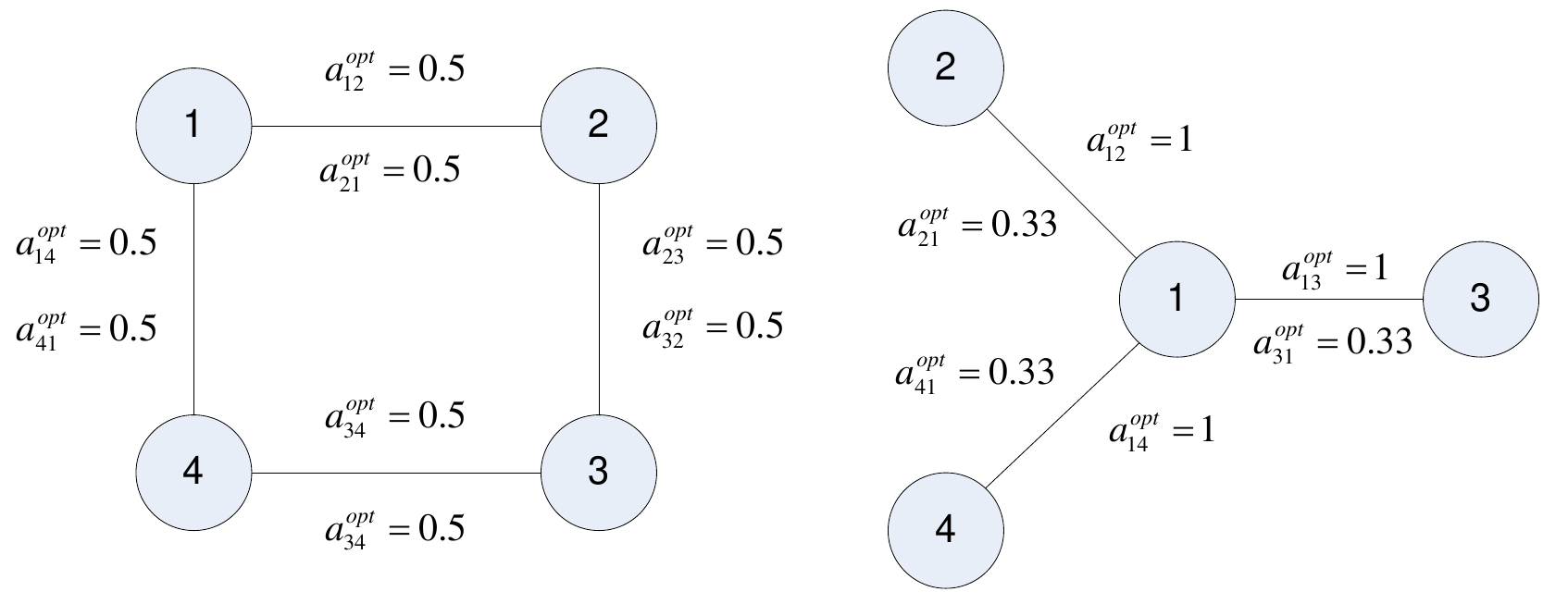}}
\caption{Optimal strategies for obedient agents interacting over a ring and a star.}\label{ring-star1}
\end{figure}

\subsection{Strategic Agents}

The information sharing problem becomes much more difficult in the presence of strategic agents: strategic agents may not want to take the prescribed actions because they do not maximize their own utilities. We formally define the network information sharing games below.

\textit{Definition 1}: A (one-shot) \textit{network information sharing game} is a tuple ${\rm {\mathcal G}}=\left\langle {\rm {\mathcal N}},{\rm {\mathcal A}},\{ u_{i} (\cdot )\} _{i\in {\rm {\mathcal N}}} ;G\right\rangle $ where ${\rm {\mathcal N}}$ is the set of players, ${\rm {\mathcal A}}$ is the action space of all players, $u_{i} (\cdot )$ is the utility function of player $i$ (defined by \eqref{ZEqnNum432082}) and $G$ is the underlying topology.

\begin{theorem}
There exists a unique Nash equilibrium (NE) $\a^{NE} =0$ in the network information sharing game in any period.
\end{theorem}
\begin{proof}
Consider the utility of an agent $i$ in \eqref{ZEqnNum432082}, the dominant
strategy of agent $i$ is $\a_{i} =0$ regardless of other agents' actions $\a_{-i} $. Therefore, the only NE is $\a_{i} =0,\forall i$.
\end{proof}

We now proceed to show how to build incentives for agents to share information with each other by exploiting their repeated interactions. In the repeated game, the (one-shot) network information sharing game is played in every period $t=0,1,2,...$. Let $y_{i}^{t} \in Y$ be the public monitoring signal related to agent $i$'s actions $\a_{i} (t)$ at time $t$. A public history of length $t$ is a sequence of public signals $(y^{0} ,y^{1} ,...,y^{t-1} )\in Y^{t} $. We note that in the considered network setting, public signals are ``locally public'' in the sense that agents only observe the public signals within their own neighborhood but not all public signals. For example, a public signal $y_{i}^{t} $ can indicate whether or not agent $i$ followed the strategy at time $t$ and only the neighbors of agent $i$ observe it.  We write ${\rm {\mathcal H}}(t)$ for the set of public histories of length $t$, ${\rm {\mathcal H}}^{T} =\bigcup _{t=0}^{T} {\rm {\mathcal H}}(t)$ for the set of public histories of length at most $T$ and ${\rm {\mathcal H}}=\bigcup _{t=0}^{\infty } {\rm {\mathcal H}}(t)$ for the set of all public histories of all finite lengths. A public strategy of agent $i$ is a mapping from public histories (in fact, only those public signals $\{ y_{j}^{t} {\rm \} }_{j:g_{ij} =1} $ that agent $i$ can observe) to $i$'s pure actions ${\it \bm\sigma }_{i} :{\rm {\mathcal H}}\to {\rm {\mathcal A}}_{i} $. We write ${\bm\sigma }$ as the collection of public strategies for all agents. Let $\delta \in (0,1]$ be the discount factor of agents. Since interactions are on-going, agents care about their long-term utilities. The long-term utility  for an agent $i$ is defined as follows:
\begin{equation} \label{5)}
U_{i} (t)=u_{i} (\a(t))+\delta u_{i} (\a(t+1))+\delta ^{2} u_{i} (\a(t+2))+...
\end{equation}
A public strategy profile ${\bm\sigma }$ induces a probability distribution over public histories and hence over\textit{ ex ante }utilities. We abuse notation and write $U_{i} ({\bm \sigma };h)$ for the expected long-run average \textit{ex ante} utility of agent $i$ when agents follow the strategy profile ${\bm \sigma }$ after the public history $h\in {\rm {\mathcal H}}$.

\textit{Definition 1}: (Perfect Public Equilibrium) A strategy profile ${\bm \sigma }$ is a perfect public equilibrium  if $\forall h\in {\mathcal H}$,$\forall i$, $U_{i} ({\it \sigma }_{i} ,{\it \sigma }_{-i} ;h)\ge U_{i} ({\it \sigma }_{i} ',{\it \sigma }_{-i} ;h),\forall {\it \sigma }_{i} '\ne {\it \sigma }_{i} $.

In the above formulation, we restrict agents to use public strategies and assume that agents make no use of any information other than provided by the (local) public signal (See Figure \ref{localpublicsignal}); in particular, agents make no use of their private history (i.e. the history sequence of its own actions $\a_{i} (t)$, its own utilities $u_{i} (t)$ and its neighbors' action toward it $\hat{\a}_{i} (t)$). This assumption admits a number of possible interpretations \cite{Mailath}, each of which is appropriate in some circumstances. In the considered scenarios where agents interact over a topology, the most important reason why we consider the design of public strategies and PPE is due to agents' partial observations and asymmetric knowledge of the network. In particular, since agent  $i$ only observes its own neighborhood subject to the underlying topology, it cannot distinguish based solely on its private history between the case in which its neighbor is deviating from the recommended strategy and the case in which its neighbor is following the recommended strategy and punishing its own neighbors' deviation actions. Using (local) public histories is more practical in the considered scenarios since it allows agents to have common knowledge within each neighborhood. The proposed rating protocols go one step further in reducing the implementation complexity by associating each agent with a rating that summarizes the public history of that agent. In this way, the space of public histories is reduced to a finite set and hence, much simpler strategies can be constructed which can still achieve the optimal social welfare.

\begin{figure}
\centerline{\includegraphics[scale = 0.8]{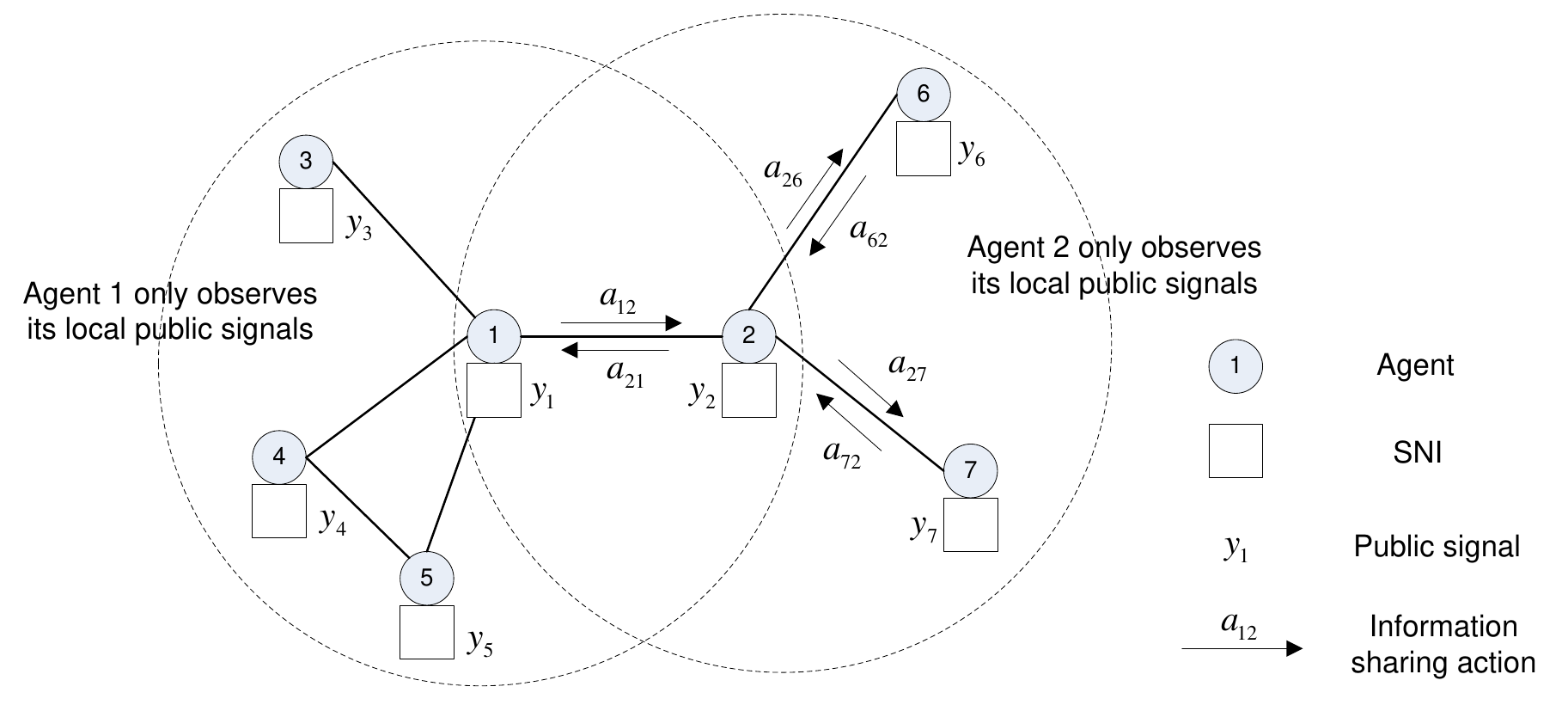}}
\caption{Illustration of local public signals. (Agents observe only public signals generated by the SNIs in their neighborhood)}\label{localpublicsignal}
\vspace{-5pt}
\end{figure}

\section{Proposed Rating Protocols}
In this section, we describe the proposed distributed rating protocol and its operation in a distributed network. As mentioned in the Introduction, each agent is equipped with an SNI. These SNIs are non-strategic software/hardware components available to the agents and will assist in the distributed design and implementation of the rating protocol. Importantly though, note that the agents \textit{are strategic} in choosing the information sharing actions (i.e. they will selfishly decide whether or not to follow the strategy recommended by the SNIs) such that their own utilities are maximized.

\subsection{Considered Rating Protocol}
A rating protocol, which is designed and implemented by the SNIs, consists of three components -- a set of ratings, a set of recommended strategies to agents, and a rating update rule.

\begin{enumerate}
\item  We consider a set of $K$ ordered ratings $\Theta =\{ 1,2,...,K\} $ with $1$ being the lowest and $K$ being the highest rating. Denote agent $i$'s rating in period $t$ by $\theta _{i} (t)\in \Theta $ and agent $i$'s neighbors' ratings by $\hat{{\bm \theta }}_{i} =\{ \theta _{j} \} _{j:g_{ij} =1} $. $K$ serves as an upper bound of the rating set size and is predetermined before the system operates.

\item  The SNIs determine the recommended (public) strategy in a distributed manner and recommend actions to their own agent depending on neighbors' ratings $\bm\sigma :{\rm {\mathcal N}}\times {\rm {\mathcal N}}\times \Theta \to [0,1]$, where $\sigma _{ij} (\theta _{j} )$ represents the recommended sharing action of agent $i$ with respect to agent $j$ if agent $j$'s rating is $\theta _{j} $. Since it is reasonable that high-rated agents should be rewarded while low-rated agents should be punished, the recommended strategy should satisfy that $\sigma _{ij} (\theta )\le \sigma _{ij} (\theta ')$ if $\theta <\theta '$. We collect the strategies of agent $i$ to all its neighbors in ${\bm \sigma }_{i} (\hat{{\bm \theta }}_{i} )=\{ \sigma _{ij} (\theta _{j} )\} _{j:g_{ij} =1} $ and the strategies of agent $i$'s neighbors to itself in $\hat{{\bm \sigma }}_{i} (\theta _{i} )=\{ \sigma _{ji} (\theta _{i} )\} _{j:g_{ij} =1} $.

\item  Depending on whether an agent $i$ followed or not the recommended strategy, the SNI of agent $i$ updates agent $i$'s rating at the end of each period.  Let $y_{i} \in Y=[0,1]$  be the monitoring signal with respect to agent $i$. Specifically, $y_{i} =1$ if $\a_{i} =\bm\sigma _{i} $ and $y_{i} =0$ if $\a_{i} \ne \bm\sigma _{i} $. The rating update rule is therefore a mapping $\tau :{\rm {\mathcal N}}\times \Theta \times Y\to \Delta (\Theta )$, where $\tau _{i} (\theta _{i}^{+} ;\theta _{i} ,y_{i} )$ is the probability that the updated rating is $\theta _{i}^{+} $ if agent $i$'s current rating is $\theta _{i} $ and the public signal is $y_{i} $. In particular, we consider the following parameterized rating update rule (see also Figure \ref{ratingupdate}), for agent $i$, if $\theta _{i} =k$,
\begin{equation} \label{6)}
\tau _{i} (\theta _{i}^{+} ;\theta _{i} ,y){\rm =}
\left\{\begin{array}{l}
{\alpha _{i,k} ,{\rm \; \; \; \;if\; \; }\theta _{i}^{+} =\max \{ 1,k-1\} ,y_{i} =0} \\ {1-\alpha _{i,k} ,{\rm \; \; \; \; if\; \; }\theta _{i}^{+} =k,y_{i} =0} \\ {\beta _{i,k} ,{\rm \; \; \; \;  if\; \; }\theta _{i}^{+} =\min \{ K,k+1\} ,y_{i} =1} \\ {1-\beta _{i,k} ,{\rm \; \; \; \; if\; \; }\theta _{i}^{+} =k,y_{i} =1}
\end{array}\right.
\end{equation}

    In words, compliant agents are rewarded to receive a higher rating with some probability while deviating agents are punished to receive a lower rating with some (other) probability. These probabilities $\alpha _{i,k} ,\beta _{i,k} $ are chosen from $[0,1]$. Note that when $\alpha _{i,k} =0$, the rating label set of agent $i$ effectively reduces to a subset $\{ k,k+1,...,K\} $ since its rating will never drop below $k$ (if its initial rating is higher than $k$). Note also that agents remain at the highest rating $\theta =K$ if they always follow the recommended strategy regardless of the choice of $\beta _{i,K} $.
\end{enumerate}

\begin{figure}
\centerline{\includegraphics[scale = 0.9]{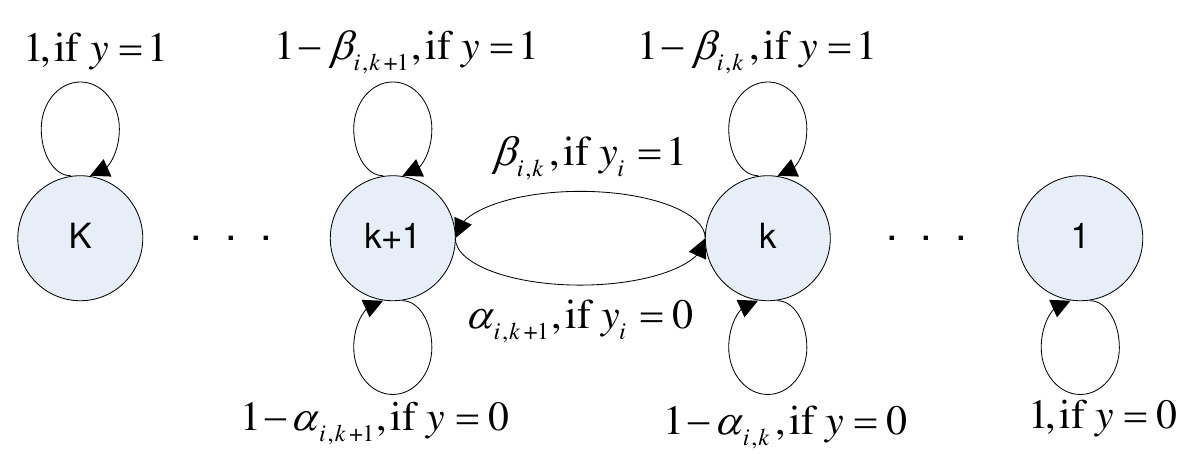}}
\caption{Rating update rule.}\label{ratingupdate}
\vspace{-15pt}
\end{figure}

Monitoring may not be perfect in implementation and hence it is possible that even if $a_{i} =\sigma _{i} $, it can still be $y_{i} =0$ (and if $a_{i} \ne \sigma _{i} $, $y_{i} =1$). If monitoring is perfect, then the strongest punishment (i.e. the agent receives the lowest rating forever once a deviation is detected) will provide the strongest incentives for agents to cooperate. However, in the imperfect monitoring environment, such punishment will lead to the network collapse where no agents share information with others. Hence, when designing the rating update rule, the monitoring errors should also be taken into account.

To sum up, the rating protocol is uniquely determined by the recommended (public) strategies ${\bm \sigma }_{i} (\hat{{\bm \theta }}_{i} ),\forall i,\forall \hat{{\bm \theta }}_{i}$ and the rating update probabilities $\alpha _{i,k} ,\beta _{i,k}, \forall i,\forall k$. We denote the rating protocol by $\pi =(\Theta ,{\bm \sigma },{\bm \alpha },{\bm \beta })$. Different rating protocols lead to different social welfare. Denote the achievable social welfare by adopting the rating protocol by $V(\pi )$. The rating protocol design problem thus is
\begin{equation}
\begin{array}{l}
{\mathop{{\rm maximize}}\limits_{\pi =(\Theta ,\bm\sigma, \bm\alpha, \bm\beta )} {\rm \; \; \; \; \; \; \; \; }V(\pi )} \\ {{\rm \;\;subject\; to\; \; \; \; \; \; \; \; \;\;}\bm\sigma {\rm \; constitutes\; a\; PPE}} \end{array}\label{ZEqnNum479148}
\end{equation}

\subsection{Operation of the Rating Protocol}
The operation of the rating protocol comprises two phases: the design phase and the implementation phase. In the design phase, the SNIs determine in a distributed way the recommended strategy and rating update rules according to the network topology, and the agents do nothing except being informed of the instantiated rating protocol. In the implementation phase (run-time), the agents (freely and selfishly) choose their actions in each information sharing period in order to maximize their own utilities (i.e. they can freely decide whether to follow or not the recommended strategies). Depending on whether the agents are following or deviating from the recommended strategy, each SNI executes the rating update of its agent and sends the new ratings of its agent to the neighboring SNIs. Note that if the rating protocol constitutes a PPE, then the agents will follow the recommended strategy in any period. When the network topology is static, the rating protocol goes through the design phase only once, when the network becomes operational, and then enters the implementation phase. When the network topology is dynamic, the rating protocol re-enters the design phase periodically, to adapt to the varying topology. However, both the design and implementation have to be carried out in a distributed way in the informationally decentralized environment. Table 2 summarizes the available information and actions of the agents and SNIs in both the design and implementation phases.

\begin{table*}[t]
\centerline{\includegraphics[scale = 0.9]{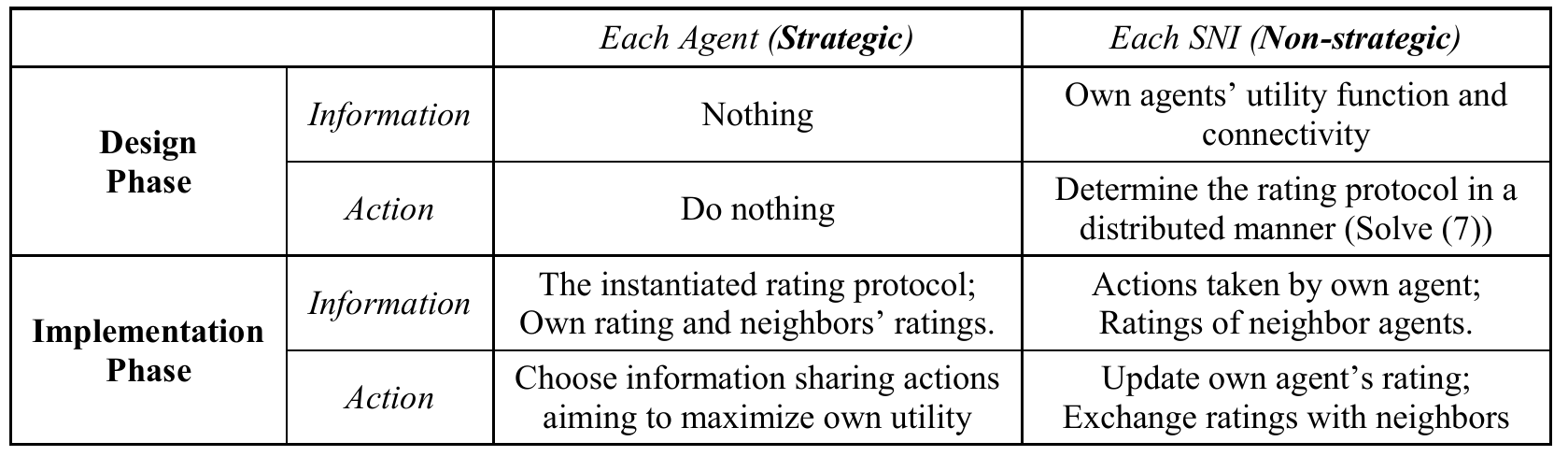}}
\caption{Operation of the rating protocol. }
\label{com_existing}
\vspace{-20pt}
\end{table*}

\section{Distributed Optimal Rating Protocol Design}
If a rating protocol constitutes a PPE, then all agents will find it in their self-interest to follow the recommended strategies. If the rating update rule updates compliant agents' to a higher rating with positive probabilities, then eventually all agents will have the highest ratings forever (assuming no update errors). Therefore, the social welfare, which is the time-average sum utilities, is asymptotically the same as the sum utilities of all agents when they have the highest ratings and follow the recommended strategy, i.e.
\begin{equation} \label{ZEqnNum559005}
V=\sum _{i}(b_{i} (\hat{{\bm \sigma }}_{i} (K))-{\bm \sigma }_{i} (\K))
\end{equation}
This means that the recommended strategies for the highest ratings determine the social welfare that can be achieved by the rating protocol. If these strategies can be arbitrarily chosen, then we can solve a similar problem as \eqref{ZEqnNum918213} for the obedient agent case. However, in the presence of self-interested agents, these strategies, together with the other components of a rating protocol, need to satisfy the equilibrium constraint such that self-interested agents have incentives to follow the recommended strategies. In Theorem 2, we identify a sufficient and necessary condition on ${\bm \sigma }(\K)$ (i.e. the recommended strategies when agents have the highest ratings) such that an equilibrium rating protocol can be constructed. With this, the SNIs are able to determine the optimal rating protocol in a distributed way in order to maximize the social welfare. We denote the social welfare that can be achieved by the optimal rating protocol as $V^{*} $ and use \textit{the price of anarchy} (PoA)\footnote{ We can also use the price of stability (PoS) as the performance measure. However, since there is a unique equilibrium given the specific rating protocol, these two measures are equivalent. }, defined as $PoA=V^{opt} /V^{*} $, as the performance measure of the rating protocol.

\subsection{Sufficient and Necessary Condition}
To see whether a rating protocol can constitute a PPE, it suffices to check whether agents can improve their long-term utilities by one-shot unilateral deviation from the recommended strategy after any history (according to the one-shot deviation principle in repeated game theory \cite{Mailath}). Since in the rating protocol, the history is summarized by the ratings, this reduces to checking the long-term utility in any state (any rating profile ${\bm \theta }$ of agents). Agent $i$'s long-term utility when agents choose the action profile $\a$ is
\begin{equation} \label{9)}
U_{i} ({\bm \theta },\a)=u_{i} ({\bm \theta },\a)+\delta \sum _{\bm\theta '} p({\bm \theta }'|{\bm \theta },\a)U_{i}^{*} ({\bm \theta }'),
\end{equation}
where $p({\bm \theta }'|{\bm \theta },\a)$ is the rating profile transition probability which can be fully determined by the rating update rule based on agents' actions and $U_{i}^{*} ({\bm \theta }')$ is the optimal value of agent $i$ at the rating profile ${\bm \theta }'$, i.e. $U_{i}^{*} ({\bm \theta }') = \max\limits_{\a_{i} } U_{i} ({\bm \theta },\a)$. PPE requires that the recommended actions for any rating profile are the optimal actions that maximize agents' long-term utilities. Before we proceed to the proof of Theorem 2, we prove the following Lemma, whose proof is deferred to online appendix \cite{onlineappendix} due to space limitation.

\smallskip
{\bf Lemma} (1) $\forall {\bm \theta }$, the optimal action of agent $i$ is either $\a_{i}^{*} ({\bm \theta })={\bf 0}$ or $\a_{i}^{*} ({\bm \theta })={\bm \sigma }_{i} (\hat{{\bm \theta }}_{i} )$.

(2) $\forall \theta _{i} $, if for $\hat{{\bm \theta }}_{i} =\K$, $\a_{i}^{*} ({\bm \theta })={\bm \sigma }_{i} (\hat{{\bm \theta }}_{i} )$, then for any other $\hat{{\bm \theta }}_{i} $, $\a_{i}^{*} ({\bm \theta })={\bm \sigma }_{i} (\hat{{\bm \theta }}_{i} )$.

(3) Let $\hat{{\bm \theta }}_{i} =\K$, suppose $\forall \theta _{i} $, $\a_{i}^{*} ({\bm \theta })={\bm \sigma }_{i} (\hat{{\bm \theta }}_{i} )$, then $\theta _{i} <\theta '_{i} {\rm \; \; \; }\Leftrightarrow {\rm \; \; }U^*_{i} (\theta _{i} ,\hat{{\bm \theta }}_{i} )\le U^*_{i} (\theta '_{i} ,\hat{{\bm \theta }}_{i} )$
\smallskip

Lemma (1) characterizes the set of possible optimal actions. That is, self-interested agents choose to either share nothing with their agents or share the recommended amount of information with their neighbors.  Lemma (2) states that if an agent has incentives to follow the recommended strategy when all its neighbors have the highest ratings, then it will also have incentives to follow the recommended strategy in all other cases. Lemma (3) shows that the optimal long-term utility of an agent is monotonic in its ratings when all its neighbors have the highest rating -- the higher the rating the larger the long-term utility the agent obtains. With these results in hand, we are ready to present and prove Theorem 2.

\begin{theorem}
Given the rating protocol structure and the network structure (topology and individual utility functions), there exists at least one PPE (of the rating protocol) if and only if $\delta b_{i} (\hat{{\bm \sigma }}_{i} (K))\ge c_i({\bm \sigma }_{i} (\K)),\forall i$.
\end{theorem}
\begin{proof}
See Appendix.
\end{proof}

\subsection{Computing the Recommended Strategy}
Theorem 2 provides a sufficient and necessary condition for the existence of a PPE with respect to the recommended strategies when agents have the highest ratings. From \eqref{ZEqnNum559005} we already know that these strategies fully determine the social welfare that can be achieved by the rating protocol. Therefore, the optimal values of ${\bm \sigma }(\K)$ can be determined by solving the following \textit{optimal recommended strategy design} problem:
\begin{equation} \label{ZEqnNum960030}
\begin{array}{l} {\mathop{{\rm maximize}}\limits_{{\bm \sigma }} {\rm \; \; \; \; \; \; \; \; }\sum _{i}(b_{i} (\hat{{\bm \sigma }}_{i} (K))-c_i({\bm \sigma }_{i} (\K))) } \\ {{\rm subject\; to\; \; \; \; \; \; \; \; }c_i({\bm \sigma }_{i} (\K))\le \delta b_{i} (\hat{{\bm \sigma }}_{i} (K)),\forall i} \end{array}
\end{equation}
where the constraint ensures that an equilibrium rating protocol can be constructed. Note that this problem implicitly depends on the network topology since both $\hat{{\bm \sigma }}_{i} (K)$ and ${\bm \sigma }_{i} (\K),\forall i$ are topology-dependent (since for each agent $i$, the strategy is only with respect to its neighbors). In this subsection, we will write ${\bm \sigma }_{i} (\K)$ as ${\bm \sigma }_{i} $ and $\hat{{\bm \sigma }}_{i} (K)$ as $\hat{{\bm \sigma }}_{i} $ to keep the notation simple.

Now, we propose a distributed algorithm to compute these recommended strategies using dual decomposition and Lagrangian relaxation. The Optimal Recommended Strategy Design problem \eqref{ZEqnNum960030} is decomposed into $N$ sub-problems each of which is solved locally by the SNIs. Note that unlike the case with obedient agents, these sub-problems have coupled constraints. Therefore, SNIs will need to go through an iterative process to exchange messages (the Lagrangians) with their neighboring SNIs such that their local solutions converge to the global optimal solution. We perform dual decomposition on \eqref{ZEqnNum960030} and relax the constraints as follows
\begin{equation} \label{ZEqnNum321261}
\mathop{{\rm maximize}}\limits_{{\bm \sigma }} {\rm \;  \; \; \; }\sum _{i}(b_{i} (\hat{{\bm \sigma }}_{i} )-\|{\bm \sigma }_{i}\| ) -\sum _{i}\lambda _{i}  (\|{\bm \sigma }_{i}\| -\delta b_{i} (\hat{{\bm \sigma }}_{i} ))
\end{equation}
where $\lambda _{i} \ge 0,\forall i$ are the Lagrangian multiplexers. The optimization thus separates into two levels of optimization. At the lower level, we have the sub-problems (one for each agent), $\forall i$
\begin{equation} \label{ZEqnNum852632}
\mathop{{\rm maximize}}\limits_{\hat{{\bm \sigma }}_{i} } {\rm \; \; \; \; \; \; \; \; }(1+\lambda _{i} \delta )b_{i} (\hat{{\bm \sigma }}_{i} )-\sum _{j:g_{ij} =1} (1+\lambda _{j} )\sigma _{ji}
\end{equation}
It is easy to see that the optimal solution of these subproblems is also the optimal solution of the relaxed problem \eqref{ZEqnNum321261}. At the higher level, the master dual problem is in charge of updating the dual variables,
\begin{equation} \label{13)}
\begin{array}{l} {\mathop{{\rm minimize}}\limits_{{\bm \lambda }} {\rm \; \; \; \; \; \; \; \; }g({\bm \lambda })=\sum _{i}g_{i}  ({\bm \lambda })} \\ {{\rm subject\; to\; \; \; \; \; \; \; \; }\lambda _{i} \ge 0,\forall i} \end{array}
\end{equation}
where $g_{i} ({\bm \lambda })$ is the maximum value of the Lagrangian \eqref{ZEqnNum852632} given ${\bm \lambda }$ and $g({\bm \lambda })$ is the maximum value of the Lagrangian \eqref{ZEqnNum321261} of the primal problem. The following subgradient method is used to update ${\bm \lambda }$,
\begin{equation} \label{ZEqnNum330405}
\lambda _{i} (q+1)=\left[\lambda _{i} (q)+w({\bm \sigma }_{i} -\delta b_{i} (\hat{{\bm \sigma }}_{i} ))\right]^{+} ,\forall i
\end{equation}
where $q$ is the iteration index, $w>0$ is a sufficiently small positive step-size. Because \eqref{ZEqnNum960030} is a convex optimization, such an iterative algorithm will converge \cite{Boyd} to the dual optimal ${\bm \lambda }^{*} $ as $q\to \infty $ and the primal variable ${\bm \sigma }^{*} ({\bm \lambda }(q))$ will also converge to the primal optimal ${\bm \sigma }^{*} $.

This iterative process can be made fully distributed which requires only limited message exchange between neighboring SNIs. We present the Distributed Computation of the Recommended Strategy (DCRS) Algorithm below which is run locally by each SNI of the agents.

\bigskip
\noindent
\begin{tabular}{p{6in}} \hline
\textbf{Algorithm}: Distributed Computation of the Recommended Strategy (DCRS)  \\ \hline
(Run by SNI of agent $i$)\textit{\newline Input}: Connectivity and utility function of agent $i$.\newline \textit{Output}:  ${\bm \sigma }_{i} (\K)=\{ \sigma _{ij} (K)\} _{j:g_{ij} =1} $ (denoted by ${\bm \sigma }_{i} =\{ \sigma _{ij} \} _{j:g_{ij} =1} $ for simplification) \\ \hline
\textbf{Initialization}:, $q=0$; $\lambda _{i} (q)=0$\newline \textbf{Repeat}:\newline Send $\lambda _{i} (q)$ to neighbor $j$, $\forall j:g_{ij} =1$. $~~~$(Obtain $\lambda _{j} (q)$ from $j$, $\forall j$)\newline Solve \eqref{ZEqnNum852632} using $\lambda _{i} (q)$, $\{\lambda _{j} (q)\}_{j:g_{ij} =1}$ to obtain $\hat{{\bm \sigma }}_{i} ({\bm \lambda }(q))$.\newline Send $\sigma _{ji} ({\bm \lambda }(q))$ to neighbor $j$, $\forall j:g_{ij} =1$.  $~~~$(Obtain $\sigma _{ij} ({\bm \lambda }(q))$ from $j$, $\forall j$)\newline Update $\lambda _{i} (q+1)$  according to \eqref{ZEqnNum330405}.\newline \textbf{Stop} until $\|\lambda _{ji} (q+1)-\lambda _{ji} (q)\|_2 <\varepsilon _{\lambda } $ \\ \hline
\end{tabular}
\bigskip

The above DCRS algorithm has the following interpretation. In each period, each SNI computes the information sharing actions of its neighbors that maximize the social surplus with respect to its own agent (i.e. the benefit obtained by its own agent minus the cost incurred by its neighbors). However, this computation has to take into account whether neighboring agents' incentive constraints are satisfied which are reflected by the Lagrangian multipliers. The larger $\lambda _{i} $ is, the more likely is that agent $i$'s incentive is being violated. Hence, the neighbors of agent $i$ should acquire less information from it. We note that the DCRS algorithm needs to be run to compute the optimal strategy only once in the static topology case or once in a while in the dynamic topology case.

\subsection{Computing the Remaining Components of the Rating Protocol}
Even though the DCRS algorithm provides a distributed way to compute the recommended strategy when agents have the highest ratings, the other elements of the rating protocol remain to be determined. There are many possible rating protocols that can constitute PPE given the obtained recommended strategies. In fact, we have already provided one way to compute these remaining elements when we determined the sufficient condition in Theorem 2 by using a constructional method. However, this is not the most efficient design in the imperfect monitoring scenario where ratings will occasionally drop due to monitoring errors. Therefore, the remaining components of the rating protocol should still be smartly chosen in the presence of monitoring errors. In this subsection, we consider a rating protocol with a binary rating set $\Theta =\{ 1,2\} $ and $\sigma _{ij} (\theta =1)=0,\forall i,j:g_{ij} =1$. We design the rating update probabilities $\alpha _{i,2} ,\beta _{i,1} ,\forall i$ to maximize the social welfare when monitoring error exists.

\begin{proposition}
Given a binary rating protocol $\Theta =\{ 1,2\} $, $\sigma _{ij} (2),\forall i,j:g_{ij} =1$ determined by the DCRS Algorithm and $\sigma _{ij}(1)=0,\forall i,j:g_{ij} =1$, when the monitoring error $\eps>0$, the optimal rating update probability that maximize the social welfare is, $\forall i$, $\beta _{i,1}^{*} =1,\alpha _{i,2}^{*} =\frac{\|{\bm \sigma }_{i} ({\bf 2})\|}{\delta b_{i} (\hat{{\bm \sigma }}_{i} (2))} $
\end{proposition}
\begin{proof}
We can derive the feasible values of $\alpha _{i,2} ,\beta _{i,1} ,\forall i$ for binary rating protocol, i.e.
\begin{equation} \label{ZEqnNum221816}
\beta _{i,1} \ge \frac{1-\delta }{\delta } \frac{\|{\bm \sigma }_{i} ({\bf 2})\|}{b_{i} (\hat{{\bm \sigma }}_{i} (2))-\|{\bm \sigma }_{i} ({\bf 2})\|}
\end{equation}
\begin{equation} \label{ZEqnNum2218162}
\alpha _{i,2} \ge \frac{1-\delta (1-\beta _{i,1} )}{\delta } \frac{\|{\bm \sigma }_{i} ({\bf 2})\|}{b_{i} (\hat{{\bm \sigma }}_{i} (2))}
\end{equation}

When monitoring is imperfect $\eps>0$, agent $i$ will drop to $\theta _{i} =1$ with positive probability even if it follows the recommended strategy all the time. According to the rating update rule, we can compute the stationary probability that agent $i$ stays at rating $\theta _{i} =2$, i.e.
\begin{equation} \label{ZEqnNum259722}
\frac{(1-\eps)\beta _{i,1} }{\eps\alpha _{i,2} +(1-\eps)\beta _{i,1} }
\end{equation}

Because agents having low ratings harms the social welfare, we need to select $\alpha _{i,2} ,\beta _{i,1} $ that maximizes \eqref{ZEqnNum259722}. This is equivalent to minimize $\alpha _{i,2} /\beta _{i,1} $.  For any $\beta _{i,1} $, the optimal value of $\alpha _{i,2} $ is the binding value of \eqref{ZEqnNum2218162} and hence, we need to minimize $[1-\delta (1-\beta _{i,1} )]/\beta _{i,1} $. Because $[1-\delta (1-\beta _{i,1} )]/\beta _{i,1} $ is decreasing in $\beta _{i,1} $, the optimal value of $\beta _{i,1} $ is $\beta _{i,1}^{*} =1$. Using \eqref{ZEqnNum221816}  again, the optimal value of $\alpha _{i,2}^{*} =\frac{\|{\bm \sigma }_{i} ({\bf 2})\|}{\delta b_{i} (\hat{{\bm \sigma }}_{i} (2))} $.          \end{proof}

It is worth noting that these probabilities can be computed locally by the SNIs of the agents which do not require any information from other agents.

\subsection{Example Revisited}
At this point, we have showed how the rating protocol can be determined in a distributed manner, given the network structure. It is time to revisit the cooperative estimation example for the ring and star topologies in order to illustrate the impact of topology on agents' incentives and recommended strategies. Figure \ref{ring-star2} illustrates the optimal recommended strategies computed using the method developed in this section for these two topologies.

\begin{figure}
\centerline{\includegraphics[scale = 0.7]{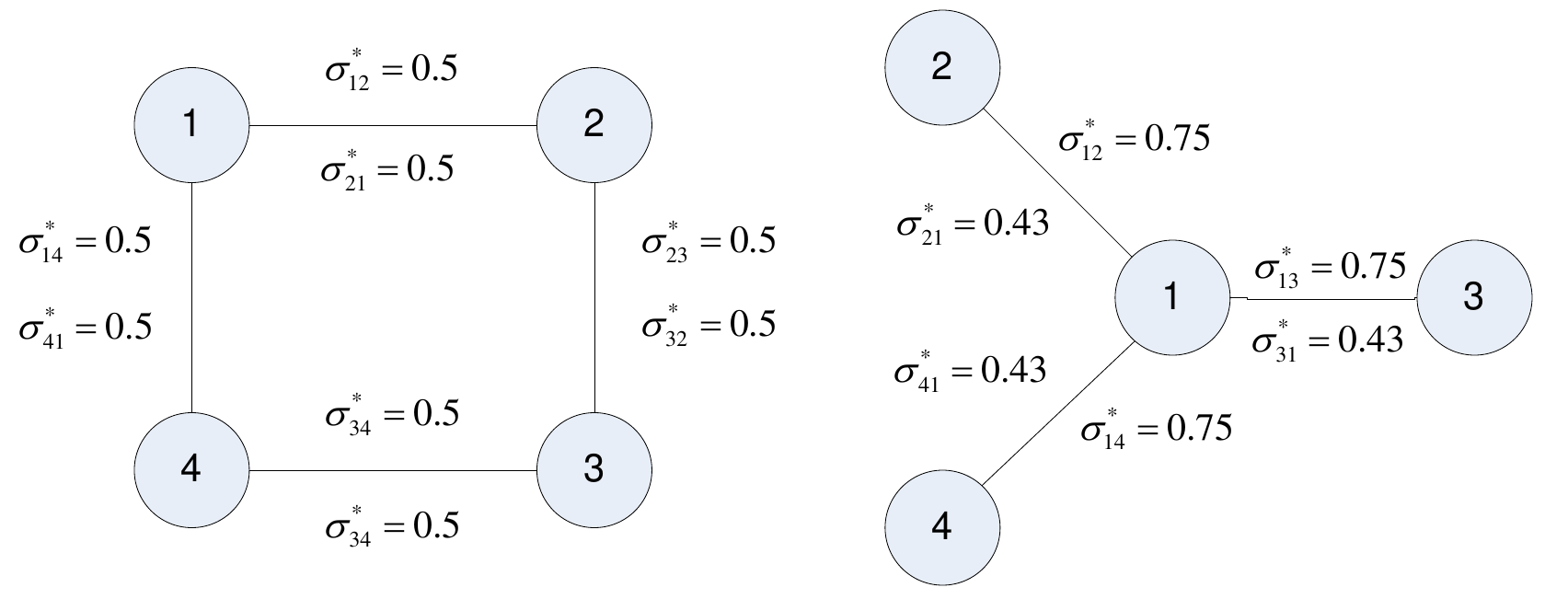}}
\caption{Optimal strategies for strategic agents interacting over a ring and a star. (The other elements of the rating protocol can be computed as in Section V(C) )}\label{ring-star2}
\vspace{-15pt}
\end{figure}

In the ring topology, agents are homogeneous and links are symmetric. As we can see, the optimal recommended strategy ${\bm \sigma }^{*} $ is exactly the same as the optimal action $\a^{opt} $ for obedient agent case because $\a^{opt} $ already provides sufficient incentive for strategic agents to follow. Therefore, we can easily determine that $PoA=1$. However, the strategic behavior of agents indeed degrades the social welfare in other cases, especially when the network becomes more heterogeneous and asymmetric, e.g. the star topologies. Even though taking $\a^{opt} $ maximizes the social welfare $V^{opt} =4$ in the star topology, these actions are not incentive-compatible for all agents. In particular, the maximum welfare $V^{opt} =4$ is achieved by sacrificing the individual utility of the center agent (i.e. agent 1 needs to contribute much more than it obtains). However, when agents are strategic, the center agent will not follow these actions $\a^{opt} $ and hence, $V^{opt} =4$ cannot be achieved. More problematically, since the center agent will choose not to participate in the information sharing process, the periphery agents do not obtain benefits and hence, they will also choose not to participate in the information sharing process. This leads to a network collapse. In the proposed rating protocol, the recommended strategies satisfy all agents' incentive constraints, namely $\delta b_{i} (\hat{{\bm \sigma }}_{i} (K))\ge \|{\bm \sigma }_{i} ({\bf K})\|,\forall i$. By comparing $\a^{opt} $ and ${\bm \sigma }^{*} $, we can see that the rating protocol recommends more information sharing from the periphery agents to the center agent and less information sharing from the central agent to the periphery agents than the obedient agent case. In this way, the center agent will obtain sufficient benefits from participating in the information sharing. However, due to this compensation for the center agent, the PoA is increased to $PoA=1.036$.

Note that the optimal recommended strategy for strategic agents is computed in a distributed way by the DCRS algorithm. Figure \ref{convergence} shows the intermediate values of the recommended strategy $\sigma _{12},\sigma _{21} $ by running the DCRS algorithm for the star. (Only the strategies between agents 1 and 2 are shown because the rest are identical due to the homogeneity of periphery agents).

\begin{figure}
\centerline{\includegraphics[scale = 0.7]{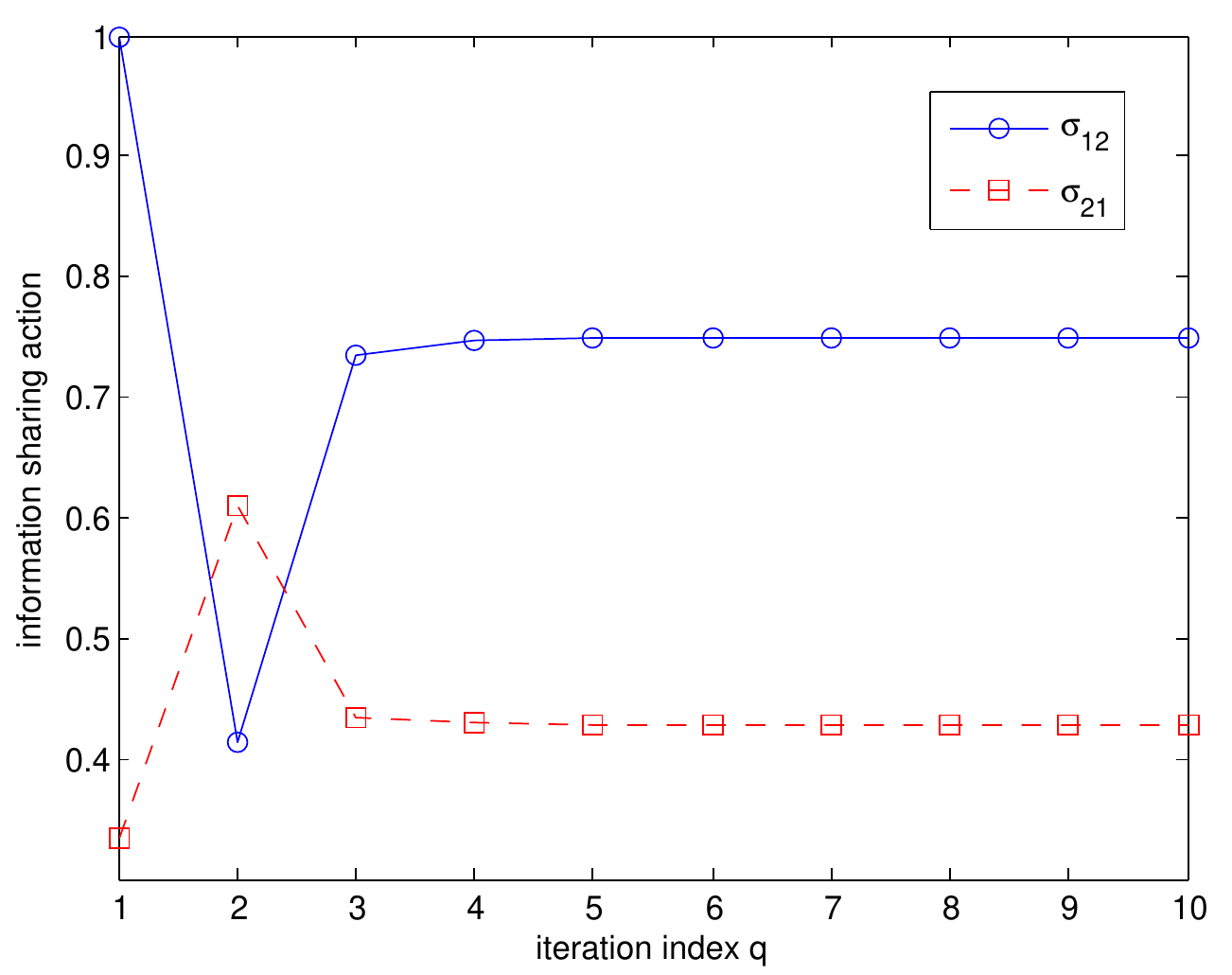}}
\caption{The recommended strategy obtained by running DCRS for the star topology.}\label{convergence}
\vspace{-15pt}
\end{figure}

\section{Performance Analysis}
In this section, we analyze the performance of the rating protocol and try to answer two important questions: (1) What is the performance loss induced by the strategic behavior of agents? (2) What is the performance improvement compared to other (simple) incentive mechanisms?

\subsection{Price of Anarchy}
Observe the social welfare maximization problems \eqref{ZEqnNum918213} and \eqref{ZEqnNum960030} for obedient agents and strategic agents (by using rating protocols), respectively. It is clear that the social welfare achieved by the rating system is always no larger than that obtained when agents are obedient due to the equilibrium constraint; hence, i.e. $PoA\ge 1$. The exact value of PoA will, in general, depend on the specific network structure (topology and individual utility functions). In this subsection, we identify a sufficient condition for the connectivity degree of the topology such that PoA is one. To simplify the analysis, we assume that agents' benefit functions are homogeneous and depend only on the sum information sharing action of the neighboring agents, i.e. $b_{i} (\hat{\a}_{i} )=b(\sum _{j:g_{ij} =1} a_{ji} )$.  Let $d_{i} =\sum _{j}g_{ij}  $ be the number of neighbors of agent $i$. The degree of network $G$ is defined as $d=\mathop{\max }\limits_{i} d_{i} $.

\begin{proposition}
Suppose benefit function structure $b_{i} (\hat{\a}_{i} )=b(\sum _{j:g_{ij} =1} a_{ji} ),\forall i$, if the connectivity degree $d$ is no larger than $\bar{d}$ such that $\delta b(\bar{d})-\bar{d}=0$, then $V^{*} =V^{opt} $, i.e. PoA is one.
\end{proposition}

\begin{proof}
Due to the concavity of the benefit function (Assumption), there exists $m^{*} $ such that if $d>m^{*} $, $b(d)-d<0$ and if $d\le m^{*} $, $b(d)-d\ge 0$. If the connectivity degree satisfies $d<m^{*} $, then the optimal solution of \eqref{ZEqnNum918213} is $a_{ij} =1,\forall i,j:g_{ij} =1$. That is, optimality is achieved when all agents share the maximal amount of information with all their neighbors. Therefore, $\forall d<m^{*} $, the agent $i$'s benefit is $b(m_{i} )$ and its cost is $m_{i} $ in the optimal solution.

Again due to the concavity of the benefit function, there exists $\bar{d}\le m^{*} $ (inequality is due to $\delta \in (0,1]$) such that if $d>\bar{d}$, $\delta b(d)-d<0$ and if $d\le \bar{d}$, $\delta b(d)-d\ge 0$.Therefore, if $d\le \bar{d}$, $\forall i$, agent $i$'s benefit and cost satisfy $\delta b(m_{i} )-m_{i} \ge 0$. This satisfies the equilibrium constraint due to Theorem 2. Therefore, the achievable social welfare is the same.            \end{proof}

Proposition 2 states that when the connectivity degree is low, the proposed rating protocol will achieve the optimal performance even when agents are strategic.

\subsection{Comparison with Direct Reciprocation}
The proposed rating protocol is not the only incentive mechanism that can incentivize agents to share information with other agents. A well-known direct reciprocation based incentive mechanism is the Tit-for-Tat strategy, which is widely adopted in many networking applications \cite{Axelrod}-\cite{Milan}. The main feature of the Tit-for-Tat strategy is that it exploits the repeated \textit{bilateral} interactions between connected\textit{ }agents, which can be utilized to incentivize agents to \textit{directly} reciprocate to each other. However, when agents do not have bilateral interests, such mechanisms fail to provide such incentives and direct reciprocity algorithms cannot be applied.

Nevertheless, even if we assume that interests are bilateral between agents, our proposed rating protocol is still guaranteed to outperform the Tit-for-Tat strategy when the utility function takes a concave form as assumed in this paper. Intuitively, because the marginal benefit from acquiring information from one neighbor is decreasing in the total number of neighbors, agents become less incentivized to cooperate when their deviation towards some neighboring agent would not affect future information acquisition from others, as is the case with the Tit-for-Tat strategy. In the following, we formally compare our proposed rating protocol with the Tit-for-Tat strategy. We assume that an agent $i$ has two sharing actions that it can choose to collaborate with its neighboring agent $j$, i.e. $\{ 0,\bar{a}_{ij} \} $ where $\bar{a}_{ij} \in (0,1]$. The Tit-for-Tat strategy prescribes the action for each agent $i$ as follows, $\forall j:g_{ij} =1$,
\begin{equation} \label{17)}
\begin{array}{l} {a_{ij} (0)=\bar{a}_{ij} } \\ {a_{ij} (t+1)=\left\{\begin{array}{l} {\bar{a}_{ij} ,{\rm \; \; \; if\; \; }a_{ji} (t)=\bar{a}_{ji} } \\ {0,{\rm \; \; \; \; \; if\; \; }a_{ji} (t)=0} \end{array}\right. ,\forall t\ge 0} \end{array}
\end{equation}

\begin{proposition}
Given the network structure and the discount factor, any action profile $\bar{a}$ that can be sustained by the Tit-for-Tat strategy can also be sustained by the rating protocol.
\end{proposition}

\begin{proof}
Consider the interactions between any pair of agents $i,j$. In the Tit-for-Tat strategy, the long-term utility of agent $i$ by following the strategy when agent $j$ played $\bar{a}_{ji} $ in the previous period is $U_{i} =\frac{\tilde{b}_{ji} (\bar{a}_{ji} )-\bar{a}_{ij} }{1-\delta } $ where $\tilde{b}_{ji} (x)=b_{i} (\hat{a}_{i} |a_{ki} =\bar{a}_{ki} ,a_{ji} =x)$. If agent $i$ deviates in the current period, Tit-for-Tat induces a continuation history $(\{ \bar{a}_{ij} ,0\} ,\{ 0,\bar{a}_{ji} \} ,\{ \bar{a}_{ij} ,0\} ...)$ where the first components are agent $i$'s actions and the second components is agent $j$'s actions.  The long-term utility of agent $i$ by one-shot deviation is thus
\begin{equation} \label{18)}
U_{i} '=\frac{\tilde{b}_{ji} (\bar{a}_{ji} )}{1-\delta ^{2} } +\delta \frac{\tilde{b}_{ji} (0)-\bar{a}_{ij} }{1-\delta ^{2} }
\end{equation}

Incentive-compatibility requires that $U_{i} \ge U_{i} '$ and therefore
\begin{equation} \label{ZEqnNum345510}
\delta (\tilde{b}_{ji} (\bar{a}_{ji} )-\tilde{b}_{ji} (0))\ge \bar{a}_{ij}
\end{equation}

Due to the concavity of the benefit function, it is easy to see that \eqref{ZEqnNum345510} leads to $\delta b_{i} (\hat{\a}_{i} )\ge \|\a_{i}\| $ which is a sufficient condition for the rating protocol to be an equilibrium.
\end{proof}

Proposition 3 proves that the social welfare achievable by the rating protocol equals or exceeds that of the Tit-for-Tat strategy, which confirms the intuitive argument before that diminishing marginal benefit from information acquisition would result in less incentives to cooperate in an environment with only direct reciprocation than in one allowing indirect reciprocation. We note that different action profiles $\bar{\a}$ will generate different social welfare. However, computing the best $\bar{\a}$ among the incentive-compatible Tit-for-Tat strategies is often intractable since \eqref{ZEqnNum345510} is a non-convex constraint. Hence, implementing the best Tit-for-Tat strategy to maximize the social welfare is often intractable. In contrast, the proposed rating protocol does not have this problem since the equilibrium constraint established in Theorem 2 is convex and hence, the optimal recommended strategy can be solved distributed by the proposed DCRS algorithm.

\section{Growing Networks}
In Section V, we designed the optimal rating protocol by assuming that the network topology is time-invariant. In practice, the social network topology can also change over time due to, e.g. new agents joining the network and new links being created. Nevertheless, our framework can easily handle such growing networks by adopting a simple extension which refreshes the rating protocol (i.e. re-computes the recommended strategy, rating update rules and re-initializes the ratings of agents) with a certain probability each period. We call this probability the refreshing rate and denote it by $\rho \in [0,1]$. When topologies are changing, the refreshing rate will also be an important design parameter of the rating protocol.

Consider that the rating protocol was refreshed at period $T$ the last time.  Denote the probability that the rating protocol is refreshed at time $T+t$ as $p(t)$. Denote the network in period $t$ by $G(t)$. We assume that in each period a number $n(t)$ of new agents join the network and stay forever. Therefore, the network topology $G(t+1)$ will be formed based on $G(t)$ and the new agents. Let $V^{*} (G;\rho )$ be the social welfare achieved by the rating protocol if the network topology is $G$ and the refreshing is set to be $\rho $. Since there are no recommended strategy and update rules concerning the new agents before the next refreshment, existing agents have no incentives to share information with the new agents and vice versa, the new agents have no incentives to share information with their neighbors. Hence, the average social welfare achieved by the rating protocol before the next refreshment is $V^{*} (G(T);\rho )$.  The optimal refreshing rate design problem is thus,
\begin{equation} \label{ZEqnNum714883}
\begin{array}{c}
\rho ^{*} =\arg \mathop{\max }\limits_{\rho } \bigg(\underbrace{{\rm {\mathbb E}}\sum _{t=0}^{\infty }p (t)\frac{1}{t+1} \sum _{\tau =0}^{t}V^{opt} (G(T+\tau )) }_{{\rm expected\; optimal\; social\; welfare}}-\underbrace{V^{*} (G(T);\rho )}_{{\rm social\; welfare\; achieved\; by\; the\; rating\; protocol}}\bigg)
\end{array}
\end{equation}
The first term in \eqref{ZEqnNum714883} is the expected optimal social welfare and the second term is the social welfare achieved by the rating protocol.

We first investigate the expected optimal social welfare. Let the social welfare variance be $\Delta _{V} (t+1)\triangleq V^{OPT} (G(t+1))-V^{OPT} (G(t))$. It is easy to see that $\Delta _{V} (t+1)\ge 0$. We assume that the expected social welfare contribution of new agents is ${\rm {\mathbb E}}(\Delta _{V} (t))=\Delta _{v} $ which is time-independent. Given the refreshing rate $\rho $, the expected time-average optimal social welfare from $T$ to the next refreshing period can be computed as
\[\begin{array}{l} {\rm {\mathbb E}}\sum _{t=0}^{\infty }p (t)\frac{1}{t+1} \sum _{\tau =0}^{t}V^{opt} (G(T+\tau )) \\={\rm {\mathbb E}}\sum _{t=0}^{\infty }\rho  (1-\rho )^{t} \frac{1}{t+1} \sum _{\tau =0}^{t}V^{opt} (G(T+\tau ))  \\ {=V^{opt} (G(T))+\sum _{t=0}^{\infty }\rho (1-\rho )^{t} \frac{1}{t+1} {\rm {\mathbb E}}\sum _{\tau =0}^{t}\Delta _{V} (T+\tau )  } \\ =V^{opt} (G(T))+\sum _{t=0}^{\infty }\rho (1-\rho )^{t} \frac{1}{t+1} \frac{t(t+1)}{2} \Delta _{V}  \\=V^{opt} (G(T))+\frac{(1-\rho )\Delta _{V} }{2\rho }  \end{array}\]

Hence, the expected optimal social welfare is decreasing in the refreshing rate $\rho $.

Next, we investigate the relation between $V^{*} (G(T);\rho )$ and $\rho $. This is established in the proposition below.

\begin{proposition}
$V^{*} (G(T);\rho )$ is non-decreasing in $\rho $.
\end{proposition}

\begin{proof}
Due to the refreshing, an agent $i$'s long-term utility becomes
\begin{equation} \label{21)}
U_{i} (t)=u_{i} (\a(t))+(1-\rho )\delta u_{i} (\a(t+1))+[(1-\rho )\delta ]^{2} u_{i} (\a(t+2))+...
\end{equation}

Hence, following the similar proof of Theorem 2, agents' incentives can be provided if and only if $(1-\rho )\delta b_{i} (\hat{{\bm \sigma }}_{i} (K))\ge {\bm \sigma }_{i} (\K),\forall i$. Therefore the constraint in the optimal strategy design problem \eqref{ZEqnNum960030} becomes stronger for the rating protocol with refreshing. Hence, the achievable social welfare becomes (weakly) lower.
\end{proof}

Summarizing, the refreshing rate impacts the social welfare gap in two different ways. On one hand, $\frac{(1-\rho )\Delta _{V} }{2\rho } $ is non-decreasing in $\rho $ since a larger $\rho $ leads to a better adaptation of the rating protocol to the changing topology. On the other hand, $V^{*} (G(T);\rho )$ is also non-decreasing in $\rho $ since a smaller $\rho $ provides more incentives for agents to follow the rating protocol designed in period $T$. Therefore, the refreshing rate has to balance these two effects. In the simulations, we will show how different refreshing rates influence the social welfare in various exemplary scenarios.

\section{Illustrative Results}
In this section, we provide simulation results to illustrate the performance of the rating protocol. In all simulations, we consider the cooperative estimation problem introduced in Section III (A). Therefore, agents' utility function takes the form of $u_{i} (\a(t))=[r^{2} -MSE_{i} (\hat{\a}_{i} (t))]-\a_{i} (t)$ \cite{Chen}. We will investigate different aspects of the rating protocol by varying the underlying topologies and the environment parameters.

\subsection{Impact of Network Topology}
Now we investigate in more detail how the agents' connectivity shapes their incentives and influences the resulting social welfare. In the first experiment, we consider the cooperative estimation over star topologies with different sizes (hence, different connectivity degrees). Figure \ref{stardegree} shows the PoA achieved by the rating protocol for discount factors $\delta =1,0.9,0.8,0.7$ for the noise variance $r^{2} =8$. As predicted by Proposition 3, when the connectivity degree is small enough, the PoA equals one and hence, the performance gap is zero. As the network size increases (hence the connectivity degree increases in the star topology), the socially optimal action requires the center agent to share more information with the periphery agents. However, it becomes more difficult for the center agent to have incentives to do so since the information sharing cost becomes much larger than the benefit. In order to provide sufficient incentives for the center agent to participate in the information sharing process, the rating protocol recommends less information sharing from the center agent to each periphery agent. However, incentives are provided at a cost of reduced social welfare. Figure \ref{stardegree} also reveals that when agents' discount factor is lower (agents value less the future utility), incentives are more difficult to provide and hence, the PoA becomes higher. In the next simulation, we study scale-free networks in the imperfect monitoring scenarios. In scale-free networks, the number of neighboring agents is distributed as a power law (denote the power law parameter by $d^{SF} $). Table 3 shows the PoA achieved by the rating protocol developed in Section V(C) for various values of $d^{SF} $ and different monitoring error probabilities $\eps$. As we can see, the proposed rating protocol achieves close-to-optimal social welfare in all the simulated environments

\begin{figure}
\centerline{\includegraphics[scale = 0.8]{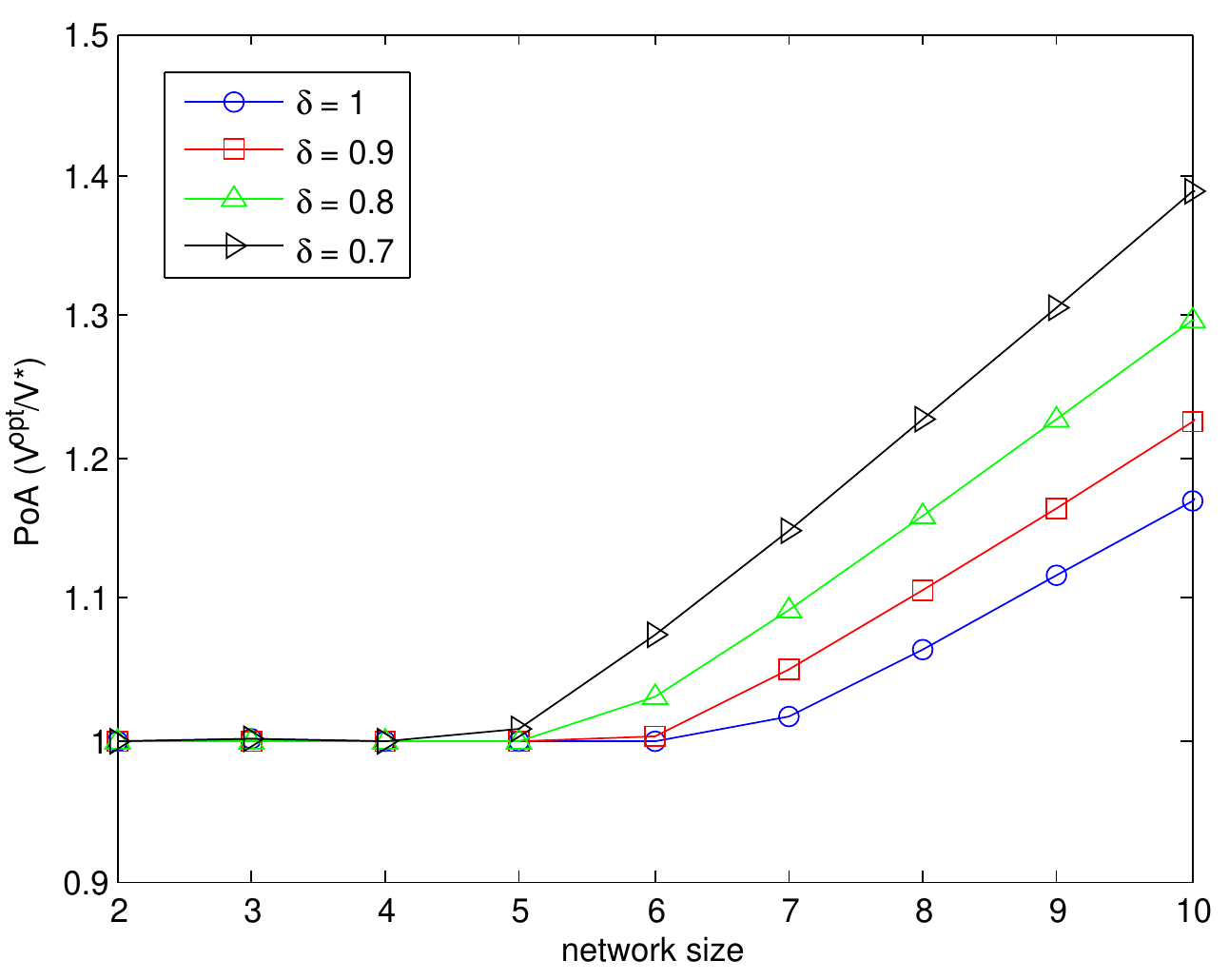}}
\caption{Performance of the rating protocol for various connectivity degrees in star topologies.}\label{stardegree}
\vspace{-5pt}
\end{figure}

\begin{table}[t]
\centerline{\includegraphics[scale = 0.9]{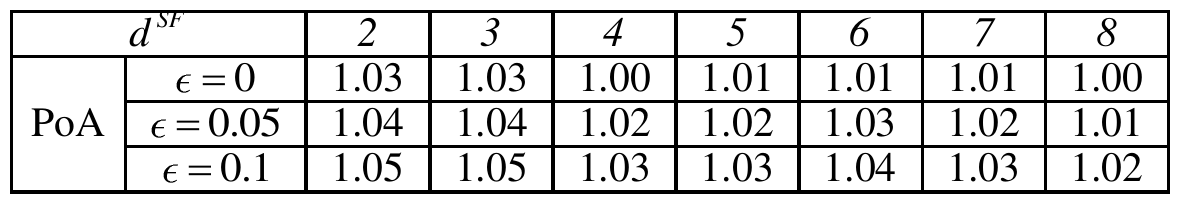}}
\caption{Performance of the rating protocol for various   in scale-free topologies.}\label{Table3}
\vspace{-15pt}
\end{table}

\subsection{Comparison with Tit-for-Tat}
As mentioned in the analysis, incentive mechanisms based on direct reciprocation such as Tit-for-Tat do not work in networks lacking bilateral interests between connected agents and hence, reasons to mutually reciprocate. In this simulation, to make possible a direct comparison with the Tit-for-Tat strategy, we consider a scenario where the connected agents do have bilateral interest and show that the proposed rating protocol significantly outperforms the Tit-for-Tat strategy. In general, computing the optimal action profile $\bar{a}^{*} $ for the Tit-for-Tat strategy is difficult because it involves the non-convex constraint $\delta (b_{i} (\{ \bar{a}_{ki}^{*} \} _{k:g_{ik} =1} )-b_{i} (\{ \bar{a}_{ki}^{*} \} _{k\ne j:g_{ik} =1} ,0))\ge \bar{a}_{ij}^{*} $, $\forall i,\forall j\ne i:g_{ij} =1$; such a difficulty is not presented in our proposed rating protocol because the constraints in our formulated problem are convex. For tractability, here we consider a symmetric and homogeneous network to enable the computation of the optimal action for the Tit-for-Tat strategy. We consider a number $N=100$ of agents and that the number of neighbors of each agent is the same $d_{i} =d,\forall i$ and each agent adopts a symmetric action profile $\bar{a}_{ij} =\bar{a},\forall i,j$. The noise variance is set to be $r^{2} =4$ in this simulation. Figure \ref{tft} illustrates the PoA achieved by the proposed rating protocol and the Tit-for-Tat strategy. As predicted by Proposition 4, any action profile that can be sustained by the Tit-for-Tat strategy can also be sustained by the proposed rating protocol (for the same $\delta $). Hence, the rating protocol yields at least as much social welfare as the Tit-for-Tat strategy. As the discount factor becomes smaller, agents' incentives to cooperate become less and hence, the PoA is larger.

\begin{figure}
\centerline{\includegraphics[scale = 0.8]{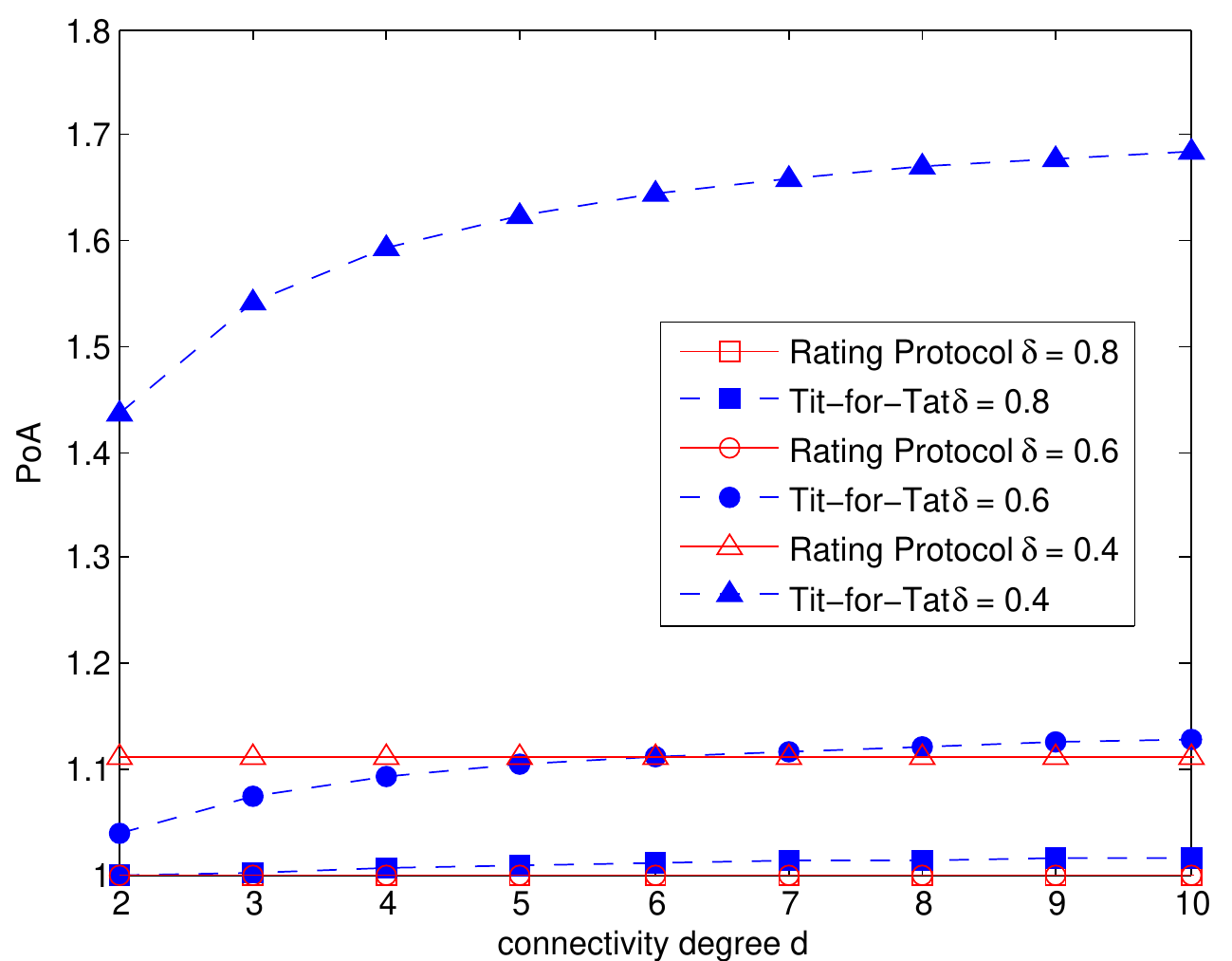}}
\caption{Performance comparison with Tit-for-Tat.}\label{tft}
\vspace{-5pt}
\end{figure}

\subsection{Rating Protocol with Refreshing}
Finally, we consider the optimal choice of the rating protocol refreshing rate $\rho $ when the network is growing as considered in section VIII. In this simulation, the network starts with $N=50$ agents. In each period, a new agent joins the network with probability 0.1 and stays in the network forever. Any two agents are connected with \textit{a priori} probability 0.2. We vary the refreshing rate from 0.005 to 0.14. Table 4 records the PoA achieved the rating protocol with refreshing for $\delta =0.4$. It shows that the optimal refreshing rate needs to be carefully chosen. If $\rho $ is too large, the incentives for agents to cooperate is small hence, the incentive-compatible rating protocol achieves less social welfare. If $\rho $ is too small, the rating protocol is not able to adapt to the changing topology well. This introduces more social welfare loss in the long-term as well. The optimal refreshing rate in the simulated network is around 0.04.

\begin{table}[t]
\centerline{\includegraphics[scale = 1]{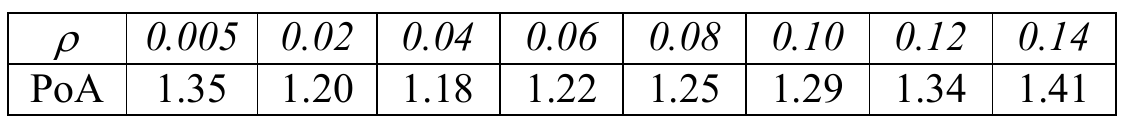}}
\caption{PoA of rating protocols with different refreshing rates. }\label{Table4}
\vspace{-15pt}
\end{table}

\section{Conclusions}
In this paper, we studied how to design distributed incentives protocols (based on ratings) aimed at maximizing the social welfare of repeated information sharing among strategic agents in social networks. We showed that it is possible to exploit the ongoing nature of agents' interactions to build incentives for agents to cooperate based on rating protocols. The proposed design framework of the rating protocol enables an efficient way to implement social reciprocity in distributed information sharing networks with arbitrary topologies and achieve much higher social welfare than existing incentive mechanisms. Our analysis also reveals the impact of different topologies on the achievable social welfare in the presence of strategic agents and hence, it provides guidelines for topology configuration and planning for networks with strategic agents. The proposed rating protocols can be applied in a wide range of applications where selfish behavior arises due to cost-benefit considerations including problems involving interactions over social networks, communications networks, power networks, transportation networks, and computer networks.

\section*{Appendix: Proof of Theorem 2}
According to Lemma, we know that it suffices to ensure that agent $i$ has the incentives to following the recommended strategy when other agents' ratings are $\K$ (i.e. all other agents have the highest rating $K$). However, we need to ensure this holds for all ratings of agent $i$. We will write ${\bm \sigma }_{i} (\K)$ as ${\bm \sigma }_{i} $ and $\hat{{\bm \sigma }}_{i} (K)$ as $\hat{{\bm \sigma }}_{i} $ to keep the notation simple.

We prove the ``only if'' part first, i.e. if $\|\bm \sigma_i\| \geq \delta b_i(\hat{\bm \sigma_i})$. Consider rating level $k$, if agent $i$ follows the recommended strategy, its long-term utility is
\begin{align}
U_i(k, \bm\sigma_i) = u_i(k, \bm\sigma_i) + \delta(\beta_{i,k} U^*_i(k+1) + (1 - \beta_{i,k} U^*_i(k))
\end{align}
By deviation to ${\bf 0}$, its long-term utility is
\begin{align}
U_i(k, {\bf 0}) = u_i(k, {\bf 0}) + \delta(\alpha_{i,k} U^*_i(k-1) + (1 - \alpha_{i,k} U^*_i(k))
\end{align}
Equilibrium requires that $U_i(k, \bm\sigma_i)  \geq U_i(k, {\bf 0})$. Hence,
\begin{equation}
\begin{aligned}
&u_i(k, {\bf 0}) - u_i(k, \bm\sigma_i)\\
\leq &\delta[(\beta_{i,k} U^*_i(k+1) + (1 - \beta_{i,k}) U^*_i(k)) \\
&- (\alpha_{i,k} U^*_i(k-1) + (1-\alpha_{i,k}) U^*_i(k))]
\end{aligned}
\end{equation}
By Lemma (3), $U^*_i(K) \geq U^*_i(k),\forall k$. Therefore, PPE requires
\begin{align}
u_i(k, {\bf 0}) - u_i(k, \bm\sigma_i) \leq \delta U^*_i(K) \label{eq34}
\end{align}
Because $u_i(k, {\bf 0}) - u_i(k, \bm\sigma_i) = \|\bm\sigma_i\|$ and
\begin{align}
U^*_i(K) = \frac{1}{1-\delta} u_i(K, \bm\sigma_i) = \frac{1}{1-\delta}\left(b_i(\hat{\bm\sigma}_i) -  \|\bm\sigma_i\|\right)
\end{align}
(\ref{eq34}) becomes,
\begin{align}
\|\bm\sigma_i\| \leq \delta b_i(\hat{\bm\sigma}_i)
\end{align}
Hence, if $\|\bm\sigma_i\| > \delta b_i(\hat{\bm\sigma}_i)$, then no rating protocol can constitute a PPE.

Next we prove the ``if'' part by construction. We let $\alpha_{i,K-1} = 0$ and hence, the effect rating set is just a binary set $\{K-1, K\}$. The value functions can be determined below,
\begin{equation}
U^*_i(K) = u_i(K, \bm\sigma_i) + \delta U^*_i(K) \label{UK}
\end{equation}
\begin{equation}
\begin{aligned}
&U^*_i(K-1) = u_i(K-1, \bm\sigma_i)\\
& + \delta (\beta_{i,K-1} U^*_i(K) + (1-\beta_{i,K-1}) U^*_i(K-1) \label{UK-1}
\end{aligned}
\end{equation}
The long-term utilities by deviation is
\begin{equation}
\begin{aligned}
&U_i(K, {\bf 0}) = u_i(K, {\bf 0}) \\
&+ \delta(\alpha_{i, K} U^*_i(K-1) + (1 - \alpha_{i, K}) U^*_i(K))
\end{aligned}
\end{equation}
\begin{equation}
U_i(K-1, {\bf 0}) = u_i(K-1, {\bf 0}) + \delta U^*_i(K-1)
\end{equation}

For agent $i$ to have incentives to following the recommended strategy at $\theta_i = K$, we need the following to hold
\begin{align}
u_i(K, {\bf 0}) - u_i(K, \bm\sigma_i) \leq \delta \alpha_{i,K}(U^*_i(K) - U^*_i(K-1)) \label{conditionK}
\end{align}

For agent $i$ to have incentives to following the recommended strategy at $\theta_i = K-1$, we need the following to hold
\begin{align}
u_i(K-1, {\bf 0}) - u_i(K-1, \bm\sigma_i) \leq \delta \beta_{i,K-1}(U^*_i(K) - U^*_i(K-1))\label{conditionK-1}
\end{align}

In the above two inequalities, $U^*_i(K) - U^*_i(K-1)$ can be computed using (\ref{UK}) and (\ref{UK-1}) and is
\begin{align}
U^*_i(K) - U^*_i(K-1) = \frac{u_i(K, \bm\sigma_i) - u_i(K-1, \bm\sigma_i)}{1 - \delta(1 - \beta_{i, K-1})}.
\end{align}

By choosing $\alpha_{i,K} = \beta_{i, K-1} = 1$, both (\ref{conditionK}) and (\ref{conditionK-1}) are satisfied. This means that if $\|\bm\sigma_i\| \leq \delta b_i(\hat{\bm\sigma}_i)$, then we can construct at least one binary rating protocol that constitutes a PPE.

%

%
%
%




\end{document}

%% file: macros.tex
\newcommand{\bp}{\begin{proof} \small }
\newcommand{\ep}{\end{proof} \normalsize}
\newcommand{\epx}{\end{proof} \small}
\newcommand{\bpa}{\begin{proofappx} \footnotesize }
\newcommand{\epa}{\end{proofappx} \small }
\newtheorem{theorem}{Theorem}
\newtheorem{proposition}{Proposition}

\newcommand{\bm}[1]{\mbox{\boldmath $#1$}}

\newcommand{\be}{\begin{equation}}
\newcommand{\ee}{\end{equation}}
\newcommand{\bs}{\begin{subequations}}
\newcommand{\es}{\end{subequations}}
\newcommand{\bq}{\begin{eqnarray}}
\newcommand{\eq}{\end{eqnarray}}
\newcommand{\bqn}{\begin{eqnarray*}}
\newcommand{\eqn}{\end{eqnarray*}}

\newcommand{\ba}{\left[ \begin{array}}
\newcommand{\ea}{\\ \end{array} \right]}
\newcommand{\ben}{\begin{enumerate}}
\newcommand{\een}{\end{enumerate}}

\def\eps {\epsilon}

\def\K{{\boldsymbol{K}}}

\def\a{{\boldsymbol{a}}}

\def\real{{\mathchoice%
{\hbox{\rm\setbox1=\hbox{I}\copy1\kern-.45\wd1 R}}
{\hbox{\rm\setbox1=\hbox{I}\copy1\kern-.45\wd1 R}}
{\hbox{\scriptsize\rm\setbox1=\hbox{I}\copy1\kern-.45\wd1 R}}
{\hbox{\scriptsize\rm\setbox1=\hbox{I}\copy1\kern-.45\wd1 R}}}}

\def\Zint{{\mathchoice{\setbox1=\hbox{\sf Z}\copy1\kern-.75\wd1\box1}
{\setbox1=\hbox{\sf Z}\copy1\kern-.75\wd1\box1}
{\setbox1=\hbox{\scriptsize\sf Z}\copy1\kern-.75\wd1\box1}
{\setbox1=\hbox{\scriptsize\sf Z}\copy1\kern-.75\wd1\box1}}}
\newcommand{\complex}{ \hbox{\rm C\kern-0.45em\rule[.07em]{.02em}{.58em}%
\kern 0.43em}}